\documentclass[a4paper,UKenglish,cleveref, autoref, thm-restate]{lipics-v2021}

\usepackage{amsthm}
\usepackage{amssymb}
\usepackage{amsmath}
\usepackage{color}
\usepackage{enumerate}
\usepackage{multirow}
\usepackage{makecell}
\usepackage{setspace}
\usepackage{graphicx}
\usepackage{tablefootnote}

\newcommand{\real} {\mathbb{R}}
\newcommand{\eps}{\varepsilon}
\newcommand{\aff}{\mathrm{aff}}
\newcommand{\cancel}[1]{}

\title{Approximate Nearest Neighbor for Polygonal Curves under Fr\'echet Distance}

\author{Siu-Wing Cheng}{Department of Computer Science and Engineering, HKUST, Hong Kong, China}{scheng@cse.ust.hk}{https://orcid.org/0000-0002-3557-9935}{}

\author{Haoqiang Huang}{Department of Computer Science and Engineering, HKUST, Hong Kong, China}{haoqiang.huang@connect.ust.hk}{https://orcid.org/0000-0003-1497-6226}{}

\authorrunning{S.-W. Cheng and H. Huang} 

\Copyright{Siu-Wing Cheng and Haoqiang Huang} 

\ccsdesc{Theory of computation~Computational geometry}\relatedversion{}\relatedversiondetails{Full Version}{https://arxiv.org/abs/2109.13460}

\keywords{Polygonal curves, Fr\'{e}chet distance, approximate nearest neighbor} 

\category{} 

\relatedversion{} 

\nolinenumbers 

\EventEditors{Uriel Feige}
\EventNoEds{2}
\EventLongTitle{50th International Colloquium on Automata, Languages and Programming (ICALP 2023)}
\EventShortTitle{ICALP 2023}
\EventAcronym{ICALP}
\EventYear{2023}
\EventDate{July 10--14, 2023}
\EventLocation{Paderborn, Germany}
\EventLogo{}
\SeriesVolume{42}
\ArticleNo{23}

\begin{document}
	
	\maketitle
	
	\begin{abstract}
		We propose $\kappa$-approximate nearest neighbor (ANN) data structures for $n$ polygonal curves under the Fr\'{e}chet distance in $\real^d$, where $\kappa \in \{1+\eps,3+\eps\}$ and $d \geq 2$.   We assume that every input curve has at most $m$ vertices, every query curve has at most $k$ vertices, $k \ll m$, and $k$ is given for preprocessing.  The query times are $\tilde{O}(k(mn)^{0.5+\eps}/\eps^d+ k(d/\eps)^{O(dk)})$ for $(1+\eps)$-ANN and $\tilde{O}(k(mn)^{0.5+\eps}/\eps^d)$ for $(3+\eps)$-ANN.  The space and expected preprocessing time are $\tilde{O}(k(mnd^d/\eps^d)^{O(k+1/\eps^2)})$ in both cases.  In two and three dimensions, we improve the query times to $O(1/\eps)^{O(k)} \cdot \tilde{O}(k)$ for $(1+\eps)$-ANN and $\tilde{O}(k)$ for $(3+\eps)$-ANN.  The space and expected preprocessing time improve to $O(mn/\eps)^{O(k)} \cdot \tilde{O}(k)$ in both cases.  For ease of presentation, we treat factors in our bounds that depend purely on $d$ as~$O(1)$.  The hidden polylog factors in the big-$\tilde{O}$ notation have powers dependent on $d$.
	\end{abstract}

	\section{Introduction}
	
	Given a set of trajectories, the \emph{nearest neighbor} problem is to efficiently report the one most similar to a query trajectory.  Trajectories are often represented as polygonal curves, and the nearest neighbor problem is encountered frequently in applications~\cite{SKS2003,SR2001,TP2002}.  
	
	Various similarity metrics have been proposed for polygonal curves.  We are interested in the \emph{Fr\'{e}chet distance}~\cite{AG1995} which has attracted much attention in recent years.  It is defined as follows.  A parameterization of a curve $\tau$ is a function $\rho: [0,1] \rightarrow \real^d$ such that, as $t$ increases from 0 to 1, the point $\rho(t)$ moves monotonically from the beginning of $\tau$ to its end.  We may have $\rho(t_1) = \rho(t_2)$ for two distinct values $t_1$ and $t_2$.   Two parameterizations $\rho$ and $\varrho$ for curves $\tau$ and $\sigma$, respectively, induce a \emph{matching} $\mathcal{M}$:  for all $t \in [0,1]$, $\mathcal{M}$ matches $\rho(t)$ with $\varrho(t)$.  A point can be matched with multiple partners.  The distance between $\tau$ and $\sigma$ under $\mathcal{M}$ is $d_\mathcal{M}(\tau,\sigma) = \max_{t \in [0,1]} d(\rho(t),\varrho(t))$, where $d(\cdot,\cdot)$ denotes the Euclidean distance.  The Fr\'{e}chet distance is $d_F(\tau,\sigma) = \min_{\mathcal{M}} d_\mathcal{M}(\tau,\sigma)$.  We call a minimizing matching a \emph{Fr\'{e}chet matching}.
	
	Let $T = \{\tau_1,\ldots,\tau_n\}$ be a set of $n$ polygonal curves with at most $m$ vertices each.  Given any value $\kappa \geq 1$,  the \emph{$\kappa$-approximate nearest neighbor (ANN) problem} is to construct a data structure so that for any query curve $\sigma$, we can quickly report a curve $\tau_l \in T$ with $d_F(\sigma,\tau_l) \leq \kappa \cdot \min_{\tau_i \in T} d_F(\sigma,\tau_i)$.   We assume that every query curve has at most $k$ vertices, and $k$ is given for preprocessing.  In the literature, if $k = m$, it is called the \emph{symmetric} version; if $k < m$, it is called the \emph{asymmetric version}.   If the query curve is sketched by the user, it is likely that $k \ll m$ and this is the scenario for which we design our data structures.  We define the related \emph{$(\kappa,\delta)$-ANN problem} as follows: for any query curve, we report ``no'' or a curve $\tau_l \in T$ with $d_F(\sigma,\tau_l) \leq \kappa\delta$; if we report `no'', it must be the case that $\min_{\tau_i \in T} d_F(\sigma,\tau_i) > \delta$.  
	
	

	There have been many results on the ANN problem under the \emph{discrete} Fr\'{e}chet distance $\tilde{d}_F$, which restricts the definition of $d_F$ to parameterizations $\rho$ and $\varrho$ that match each vertex of $\tau$ with at least one vertex of $\sigma$, and vice versa.  
	As a result, $d_F(\tau,\sigma) \leq \tilde{d}_F(\tau,\sigma)$.  It is possible that $d_F(\tau,\sigma) \ll \tilde{d}_F(\tau,\sigma)$; for example, $\sigma$ is a long horizontal line segment, and $\tau$ is a parallel copy near $\sigma$ with an extra vertex in the middle.  The advantage of $\tilde{d}_F$ is that it can be computed using a simple dynamic programming algorithm~\cite{EM1994}.
	
	
	Indyk and Motwani~\cite{IM1998} and Har-Peled~\cite{H2001} proved that a solution for the $(\kappa,\delta)$-ANN problem for points in a metric space gives a solution for the $\kappa(1+O(\eps))$-ANN problem.  The result has been simplified in the journal version~\cite{har2012approximate}.  The  method is general enough that it works for polygonal curves under $d_F$ and $\tilde{d}_F$.  Theorem~\ref{thm:reduce} in Section~\ref{sec:ANN} states the deterministic result in our context; the reduction increases the space and query time by polylogarithmic factors.  If a probabilistic $(\kappa,\delta)$-ANN solution with failure probability $f$ is used, the bounds in Theorem~\ref{thm:reduce} also hold, and the ANN solution has an $O(f\log n)$ failure probability.
	
	Indyk~\cite{indyk2002approximate} proposed the first $(\kappa,\delta)$-ANN solution under $\tilde{d}_F$, where $\kappa = O(\log m+\log\log n)$, for the case that $k = m$ and the vertices come from a discrete point set $X$.  It uses $O\bigl(|X|^{\sqrt{m}}(m^{\sqrt{m}}n)^2\bigr)$ space and answers a query in $O\big(m^{O(1)}\log n\big)$ time.\footnote{A tradeoff is also presented in~\cite{indyk2002approximate}.}  Driemel and Silverstri~\cite{driemel2017locality} developed probabilistic $(\kappa,\delta)$-ANN solutions under $\tilde{d}_F$ with a failure probability $1/n$; they achieve the following combinations of ($\kappa$, query time, space) for the case of $k = m$: $\bigl(4d^{3/2}m, O(m), O(n\log n + mn)\bigr)$, $\bigl(4d^{3/2}, O(2^{4dm}m\log n), O(2^{4md}n\log n + mn)\bigr)$, and $\bigl(4d^{3/2}m/t, O(2^{2t}m^t\log n), O(2^{2t}m^{t-1}n\log n + mn)\bigr)$ for any integer $t \geq 1$.  The approximation ratio has been reduced to $1+\eps$ by two research groups later. Filtser~et.~al.~\cite{filtser2020approximate} proposed two deterministic $(1+\eps,\delta)$-ANN data structures under $\tilde{d}_F$; one answers a query in $O(kd)$ time and uses $n \cdot O(\frac{1}{\eps})^{kd}$ space and $O(mn (d\log  m + O(\frac{1}{\eps})^{kd}))$ expected preprocessing time;  the other answers a query in $O(kd\log\frac{dkn}{\eps})$ time and uses $n\cdot O(\frac{1}{\eps})^{kd}$ space  and $O(mn\log \frac{n}{\eps}\cdot (d\log m + O(\frac{1}{\eps})^{kd}))$ worst-case preprocessing time.  Emiris and Psarros~\cite{emiris2020products} obtained probabilistic $(1+\eps)$-ANN and $(1+\eps,\delta)$-ANN data structures under $\tilde{d}_F$ with failure probabilities $1/2$ for the case of $k = m$.  The $(1+\eps)$-ANN data structure answers a query in $\tilde{O}(d2^{4m}m^{O(1/\eps)})$ time and uses $\tilde{O}(dm^2 n ) \cdot (2 + d/\log m)^{O(dm^{O(1/\eps)}\log(1/\eps))}$ space and preprocessing time.  The $(1+\eps,\delta)$-ANN data structure answers a query in $O(d2^{4m}\log n)$ time and uses $O(dn)  + (mn)^{O(m/\eps^2)}$ space and preprocessing time.  The failure probabilities can be reduced to $1/n$ with an increase in the query time, space, and preprocessing time by an $O(\log n)$ factor.
	
	Most known results under $d_F$ are for $\real$.  For curves in $\real$ (time series), Driemel and Psarros~\cite{DP2020} developed the first $(\kappa,\delta)$-ANN data structures under $d_F$ with the following combinations of ($\kappa$, query time, space): $\bigl(5+\eps,O(k),O(mn)+n\cdot O(\frac{1}{\eps})^{k}\bigr)$, $\bigl(2+\eps,O(2^kk), O(mn)+ n \cdot O(\frac{m}{k\eps})^{k}\bigr)$, and $\bigl(24k+1, O(k\log n), O(n\log n + mn)\bigr)$.  The last one is probabilistic, and the failure probability is $1/\mathrm{poly}(n)$.
	Later, Bringman~et~al.~\cite{BDNP2022} obtained improved solutions with the following combinations of ($\kappa$, query time, space): $\bigl(1+\eps, O(2^kk), n\cdot O(\frac{m}{k\eps})^{k}\bigr)$, $(2+\eps, O(k), n\cdot O(\frac{m}{k\eps})^k\bigr)$, $\bigl(2+\eps, O(2^kk), O(mn) + n \cdot O(\frac{1}{\eps})^k \bigr)$, $\bigl(2+\eps,O(\frac{1}{\eps})^{k+2}, O(mn)\bigr)$, and $\bigl(3+\eps, O(k), O(mn) + n\cdot O(\frac{1}{\eps})^k\bigr)$.  They also obtained lower bounds that are conditioned on the Orthogonal Vectors Hypothesis: for all $\eps,\eps' \in (0,1)$, it is impossible to achieve the combination $\bigl(2-\eps, O(n^{1-\eps'}),\mathrm{poly}(n)\bigr)$ in $\real$ when $1 \ll k \ll \log n$ and $m =kn^{\Theta(1/k)}$, or $\bigl(3-\eps, O(n^{1-\eps'}),\mathrm{poly}(n)\bigr)$ in $\real$ when $m = k = \Theta(\log  n)$,  or $\bigl(3-\eps, O(n^{1-\eps'}),\mathrm{poly}(n)\bigr)$ in $\real^2$ when $1 \ll k \ll \log n$ and $m = kn^{\Theta(1/k)}$.  Mirzanezhad~\cite{mirzanezhad2020approximate} described a $(1+\eps,\delta)$-ANN data structure for $\real^d$ that answer a query in $O(kd)$ time and uses $O(n \cdot \max\{\sqrt{d}/\eps,\sqrt{d}D/\eps^2\}^{dk})$ space, where $D$ is the diameter of the set of input curves.  If $k$ is not given, the approximation ratio and space increase to $5+\eps$ and $n \cdot O(\frac{1}{\eps})^{dm}$, respectively.  There is no bound on $D$ in the space complexity of the first solution. We summarize all these previous results in Table~\ref{Table:ANN} for easier comparison. 
	\begin{table}[h]

		\centering
		\begin{minipage}{\textwidth}
			
		\resizebox{\columnwidth}{!}{%
			\begin{tabular}{|c|c|c|c|}
				\hline
				Distance & Space & Query time & Approximation \\ [0.5ex]
				\hline\hline
				\multirow{8}{11em}{\centering Continuous Fr\'echet, $\real$} & $O\left(mn\right)+n\cdot O\left(\frac{1}{\eps}\right)^k$ & $O(k)$ & $(5+\epsilon, \delta)$-ANN~\cite{DP2020}\\
				&$O\left(mn\right)+n\cdot O\left(\frac{m}{k\eps}\right)^k$ & $O(2^kk)$ & $(2+\epsilon,\delta)$-ANN~\cite{DP2020}\\
				&$O\left(n\log n + mn\right)$ & $O(k\log n)$ & $(24k+1, \delta)$-ANN~\cite{DP2020}\footnote{A randomized data structure with a failure probability of $1/\text{poly}(n)$.}\\
				&$n\cdot O\left(\frac{m}{k\eps}\right)^k$ & $O(2^kk)$ & $(1+\eps, \delta)$-ANN~\cite{BDNP2022}\\
				&$n\cdot O\left(\frac{m}{k\eps}\right)^k$ & $O(k)$ & $(2+\eps, \delta)$-ANN~\cite{BDNP2022}\\
				&$O\left(mn\right)+n\cdot O\left(\frac{1}{\eps}\right)^k$ & $O(2^kk)$ & $(2+\eps, \delta)$-ANN~\cite{BDNP2022}\\
				&$O\left(mn\right)$ & $O\left(\frac{1}{\eps}\right)^{k+2}$ & $(2+\eps, \delta)$-ANN~\cite{BDNP2022}\\
				&$O\left(mn\right)+n\cdot O\left(\frac{1}{\eps}\right)^k$ & $O(k)$ & $(3+\eps, \delta)$-ANN~\cite{BDNP2022}\\
				\hline
				\multirow{4}{11em}{\centering Continuous Fr\'echet, $\real^d$} &
				$O\left(n\cdot \max\{\sqrt{d}/\eps, \sqrt{d}D/\eps^2\}^{dk}\right)$ & $O(kd)$ & $(1+\eps, \delta)$-ANN~\cite{mirzanezhad2020approximate}\\
				&$n\cdot O\left(\frac{1}{\eps}\right)^{dm}$ & $O(kd)$ & $(5+\eps, \delta)$-ANN~\cite{mirzanezhad2020approximate}\\
				&$\tilde{O}\left(k(mnd^d/\eps^d)^{O(k+1/\eps^2)}\right)$ & $\tilde{O}\left(k(mn)^{0.5+\eps}/\eps^d+k(d/\eps)^{O(dk)}\right)$ & $(1+\eps, \delta)$-ANN, Theorem~\ref{thm:ann23}\\
				&$\tilde{O}\left(k(mnd^d/\eps^d)^{O(k+1/\eps^2)}\right)$ & $\tilde{O}\left(k(mn)^{0.5+\eps}/\eps^d\right)$ & $(3+\eps, \delta)$-ANN, Theorem~\ref{thm:3ANN}\\
				\hline
				\multirow{2}{11em}{\centering Continuous Fr\'echet, $\real^2$ and $\real^3$} &
				$O\left(\frac{1}{\eps}\right)^{4d(k-1)+1}(mn)^{4(k-1)}k\log^2n$ & $O\left(\frac{1}{\eps}^{2d(k-2)}\right)k\log\frac{mn}{\eps}\log n$ & $(1+\eps, \delta)$-ANN, Theorem~\ref{thm:ann23}\\
				&$O\left(\frac{1}{\eps}\right)^{2d(k-1)+1}(mn)^{2(k-1)}k\log^2n$ & $O\left(k\log\frac{mn}{\eps}\log n\right)$ & $(3+\eps, \delta)$-ANN, Theorem~\ref{thm:3ANN}\\
				\hline
				\multirow{6}{11em}{\centering Discrete Fr\'echet, $\real^d$} & $O\left(|X|^{\sqrt{m}}(m^{\sqrt{m}}n)^2\right)$ & $O\left(m^{O(1)}\log n\right)$ &  $(O(\log m +\log\log n),\delta$)-ANN~\cite{indyk2002approximate}\\
				&$O(n\log n + mn)$& $O(m)$& $(4d^{3/2}m, \delta)$-ANN~\cite{driemel2017locality}\\
				&$O(2^{4md}n\log n + mn)$&$O(2^{4dm}m\log n)$&$(4d^{3/2}, \delta)$-ANN~\cite{driemel2017locality}\\
				&$O(2^{2t}m^{t-1}n\log n + mn)$&$O(2^{2t}m^t\log n)$&$(4d^{3/2}m/t, \delta)$-ANN~\cite{driemel2017locality}\\
				& $n\cdot O\left(\frac{1}{\eps}\right)^{kd}$ & $O(kd)$\footnote{The query time is achieved by implementing the dictionary with a hash table. The query time is $O(kd\log\frac{dkn}{\eps})$ when implementing the dictionary with a trie.} & $(1+\eps, \delta)$-ANN~\cite{filtser2020approximate}\\
				&$O(dn)+(mn)^{O(m/\eps^2)}$& $O\left(d2^{4m}\log n\right)$ & $(1+\eps, \delta)$-ANN\footnote{A randomized data structure with a failure probability of $\frac{1}{2}.$}~\cite{emiris2020products} \\
				\hline
			\end{tabular}%
		}
	
\end{minipage}
	 
		\caption{Comparison of our data structures to the previous results.}\label{Table:ANN}
	
	\end{table}

	We develop $(\kappa,\delta)$-ANN data structures under $d_F$ in $\real^d$ for $\kappa \in \{1+\eps,3+\eps\}$ and $d \geq 2$.  We assume that every query curve has at most $k$ vertices, $k \ll m$, and $k$ is given for preprocessing.   To simplify the bounds, we assume that $k \geq 3$ throughout this paper.  There are three design goals.  First, the query times are sublinear in $mn$.  Second, the space complexities depend only on the input parameters.  Third, the space complexities are neither proportional to $\min\{m^{\Omega(d)}, n^{\Omega(d)}\}$ nor exponential in $\min\{m, n\}$.  It would be desirable to remove all exponential dependencies on $d$, but we are not there yet.

	
	We achieve a query time of $\tilde{O}(k(mn)^{0.5+\eps}/\eps^{d} + k(d/\eps)^{O(dk)})$ for $(1+\eps,\delta)$-ANN.   We remove the exponential dependence on $k$ for $(3+\eps,\delta)$-ANN and obtain an $\tilde{O}(k(mn)^{0.5+\eps}/\eps^{d})$ query time.  The space and expected preprocessing time are $\tilde{O}(k(mnd^d/\eps^d)^{O(k+1/\eps^2)})$ in both cases.  
	For ease of presentation, we treat any factor in our bounds that depends only on $d$ as $O(1)$.  The hidden polylog factors in the big-$\tilde{O}$ notation have powers dependent on $d$.   
	In two and three dimensions, we improve the query times to $O(1/\eps)^{O(k)} \cdot \tilde{O}(k)$ for $(1+\eps,\delta)$-ANN and $\tilde{O}(k)$ for $(3+\eps,\delta)$-ANN.  The space and expected preprocessing time improve to $O(mn/\eps)^{O(k)} \cdot \tilde{O}(k)$ in both cases.   	Using the reduction in~\cite{har2012approximate} (Theorem~\ref{thm:reduce} in Section~\ref{sec:ANN}), we obtain $(1+\eps)$-ANN and $(3+\eps)$-ANN data structures by increasing the query time and space by an $O(\log n)$ and an $O(\frac{1}{\eps}\log^2 n)$ factors, respectively.  
	More precise bounds are stated in Theorems~\ref{thm:ann23} and~\ref{thm:3ANN}.   
	
	Our $(1+\eps,\delta)$-ANN result is based on two new ideas.  First, we develop a novel encoding of query curves that are based on local grids in the input vertex neighborhoods.  Second, we draw a connection to an approximate segment shooting problem which we solve efficiently.  We present these ideas in Sections~\ref{sec:ANN} and~\ref{Sec: subroutines}.  Our $(3+\eps)$-ANN result is obtained by simplifying the encoding.  We present this result in Section~\ref{sec:3ANN}.
	
	We work in the word RAM model.  We use $(v_{i,1},\ldots,v_{i,m})$ to denote the sequence of vertices of $\tau_i$ from beginning to end---$\tau_i$ is oriented from $v_{i,1}$ to $v_{i,m}$.  We use $\tau_{i,a}$ to denote the edge $v_{i,a}v_{i,a+1}$.    For any two points $x,y \in \tau_i$, we say that $x \leq_{\tau_i} y$ if $x$ does not appear behind $y$ along $\tau_i$, and $\tau_i[x,y]$ denotes the subcurve between $x$ and $y$.  Given two subsets $X, Y \subseteq \tau_i$, $X \leq_{\tau_i} Y$ if and only if for every point $x \in X$ and every point $y \in Y$, $x \leq_{\tau_i} y$.  A ball centered at the origin with radius $r$ is denoted by $B_r$. Given two subsets $X, Y \subset \mathbb{R}^d$, $d(X,Y) = \min_{x\in X, y\in Y} d(x,y)$; their \emph{Minkowski sum} is $X\oplus Y= \{x+y: x\in X, y\in Y\}$; if $X = \{p\}$, we write $p \oplus Y$ for simplicity.  For any $x,y \in \real^d$, $xy$ denotes the \emph{oriented segment from $x$ to $y$}, and $\aff(xy)$ is the \emph{oriented support line of $xy$} that shares the orientation of $xy$.

	\section{$\pmb{(1+O(\eps),\delta)}$-ANN}
	\label{sec:ANN}

	Har-Peled~et~al.~\cite[Theorem 2.10]{har2012approximate} proved a reduction from the $(1+\eps)$-ANN problem to the $(1+\eps,\delta)$-ANN problem.  Although the result is described for points in a metric space with a probabilistic data structure for the $(1+\eps,\delta)$-ANN problem, the method is general enough to work for polygonal curves under $d_F$ or $\tilde{d}_F$ in $\real^d$ and any deterministic solution for the $(1+\eps,\delta)$-ANN problem.  We rephrase their result in our context below.
	
	\begin{theorem}[\cite{har2012approximate}]
		\label{thm:reduce}
		Let $T$ be a set of $n$ polygonal curves in $\real^d$.  If there is a data structure for the $(\kappa,\delta)$-ANN problem for $T$ under $d_F$ or $\tilde{d}_F$ that has space complexity~$S$, query time $Q$, deletion time $D$, and preprocessing time $P$, then there is a $\kappa(1+O(\eps))$-ANN data structure for $T$ under $d_F$ or $\tilde{d}_F$ that has space complexity $O(\frac{1}{\eps}S\log^2 n)$, query time $O(Q\log n)$, and expected preprocessing time $O\bigl(\frac{1}{\eps\log^2 n}P + (Q + D)n\log n\bigr)$.
	\end{theorem}
	
	By Theorem~\ref{thm:reduce}, we can focus on the $(1+\eps,\delta)$-ANN problem.  Without loss of generality, we assume that each curve in $T$ has exactly $m$ vertices, and every query curve has exactly $k$ vertices.  If necessary, extra vertices can be added in an arbitrary manner to enforce this assumption without affecting the Fr\'{e}chet distance.
	
	 The high level idea of our preprocessing is to identify all query curves that are within a Fr\'{e}chet distance $(1+O(\eps))\delta$ from each $\tau_i \in T$, group the curves that share similar structural characteristics, assign each group a unique key value, and store these key values in a trie $\mathcal{D}$.  It is possible for a query curve to belong to multiple groups.  Each key value is associated with the subset of curves in $T$ that induce that key value.  Correspondingly, given a query curve $\sigma$, we generate all possible key values for $\sigma$ and search $\mathcal{D}$ with them.  If some curve in $T$ is retrieved, it is the desired answer; otherwise, we report ``no''.
	
	There are two challenges to overcome.  First, it is impossible to examine all possible query curves.  We can only check some space discretization in order to obtain a finite running time.  To control the discretization error, it is easy to cover the input curves by a grid with an appropriate cell width; however, the grid size and hence the data structure size would then depend on some non-combinatorial parameters.  We propose \emph{coarse encodings} of query curves so that there are $O(\sqrt{d}/\eps)^{4d(k-1)}(mn)^{4(k-1)}$ of them.  A query curve may have $O(\sqrt{d}/\eps)^{2d(k-2)}$ coarse encodings.  The second challenge is to efficiently generate all possible coarse encodings of a query curve at query time.  We reduce the coarse encoding generation to an approximate segment shooting problem.  This step turns out to be the bottleneck in four and higher dimensions as we aim to avoid any factor of the form $m^{\Omega(d)}$ or $n^{\Omega(d)}$ in the space complexity. It is the reason for the $(mn)^{0.5+\eps}$ term in the query time.  In two and three dimensions, the approximate segment shooting problem can be solved more efficiently.
	
	In the rest of this section, we present the coarse encoding and a $(1+O(\eps),\delta)$-ANN data structure, using an approximate segment shooting oracle.  The approximate segment shooting problem can be solved by the results in~\cite{BHO1994} in two and three dimensions.  We solve the approximate segment shooting problem in four and higher dimensions in Section~\ref{Sec: subroutines}.

	\subsection{Coarse encodings of query curves}
	\label{sec:coarse}
	
	Imagine an infinite grid in $\real^d$ of cell width $\varepsilon\delta/\sqrt{d}$.  For any subset $R \subset \real^d$, we use $G(R)$ to denote the set of grid cells that intersect $R$.  Let $\mathcal{G}_1 = \bigcup_{i\in[n], a\in [m]}G(v_{i,a}\oplus B_{\delta})$.  Let $\mathcal{G}_2 = \bigcup_{i\in[n], a\in [m]}G(v_{i,a}\oplus B_{(2+12\eps)\delta})$. Both $\mathcal{G}_1$ and $\mathcal{G}_2$ have $O(mn/\eps^d)$ size.

	The coarse encoding of a curve $\sigma = (w_1,w_2,\ldots,w_k)$ is a 3-tuple $\mathcal{F}=(\mathcal{A},\mathcal{B},\mathcal{C})$.
	%
	The component $\mathcal{C}$ is sequence of pairs of grid cells $((c_{j,1}, c_{j,2}))_{j\in [k-1]}$ such that $(c_{j,1}, c_{j,2}) \in  (\mathcal{G}_1 \times \mathcal{G}_1) \cup \{\text{null}\}$.  
	Both $\mathcal{A}$ and $\mathcal{B}$ are arrays of length $k-1$,
	and every element of $\mathcal{A}$ and $\mathcal{B}$ belongs to $\mathcal{G}_2 \cup \{\text{null}\}$.
	We first provide the intuition behind the design of $(\mathcal{A},\mathcal{B},\mathcal{C})$ before describing the constraints that realize the intuition.
	
	Imagine that a curve $\tau_i \in T$ is a $(1+O(\eps))$-ANN of $\sigma$.  The data structure needs to cater for the preprocessing, during which the query curve $\sigma$ is not available; it also needs to cater for the query procedure, during which we do not want to directly consult the input curves in $T$ in order to avoid a linear dependence in $mn$.
	
	In preprocessing, we use pairs of grid cells as surrogates of the possible query curve edges.  The advantage is that we can enumerate all possible pairs of grid cells and hence cater for all possible query curve edges.  Specifically, for $j \in [k-1]$, if $(c_{j,1},c_{j,2}) \not= \text{null}$, it is the surrogate of $w_jw_{j+1}$, so $w_jw_{j+1}$ should pass near $c_{j,1}$ and $c_{j,2}$.  Each non-null $(c_{j,1},c_{j,2})$ corresponds to a contiguous subsequence $v_{i,a},\ldots,v_{i,b}$ of vertices of $\tau_i$ that are matched to points in $w_jw_{j+1}$ in a Fr\'{e}chet matching.   Of course, we do not know the Fr\'{e}chet matching, so we will need to enumerate and handle all possibilities.  Also, since $w_jw_{j+1}$ is unknown in preprocessing, $v_{i,a},\ldots,v_{i,b}$ can only be matched to a segment joining a vertex $x_j$ of $c_{j,1}$ to a vertex $y_j$ of $c_{j,2}$ so that $d_F(x_j'y_j',\tau_i[v_{i,a},v_{i,b}]) \leq (1+O(\eps))\delta$ for some subsegment $x_j'y_j' \subseteq x_jy_j$.  This property will be enforced in the data structure construction later.
	
	At query time, given $\sigma = (w_1,\ldots,w_k)$, we will make approximate segment shooting queries to determine a sequence of cell pairs $((c_{j,1},c_{j,2}))_{j \in [k-1]}$.  We do not always use $(c_{j,1},c_{j,2})$ as a surrogate for the edge $w_jw_{j+1}$ though.  As mentioned in the previous paragraph, a non-null $(c_{j,1},c_{j,2})$ denotes the matching of a contiguous subsequence of input curve vertices to $w_jw_{j+1}$; however, we must also allow the matching of a contiguous subsequence of vertices of $\sigma$ to a single input edge.  Therefore, after determining $((c_{j,1},c_{j,2}))_{j \in [k-1]}$, we still have the choice of using $(c_{j,1},c_{j,2})$ as is or substituting it by the null value.  For a technical reason, $(c_{1,1},c_{1,2})$ and $(c_{k-1,1},c_{k-1,2})$ are always kept non-null, so we have $2^{k-3}$ possible sequences of pairs of cells.  Take one of these sequences.  If $(c_{r,1},c_{r,2})$ and $(c_{s,1},c_{s,2})$ are two non-null pairs such that $(c_{j,1},c_{j,2}) = \text{null}$ for $j \in [r+1,s-1]$, it means that no vertex of $\tau_i$ is matched to $w_jw_{j+1}$ for $j \in [r+1,s-1]$.  As a result, the vertices $w_{r+1},\ldots,w_s$ of $\sigma$ are matched to the edge $v_{i,b}v_{i,b+1}$ of $\tau_i$, where $v_{i,b}$ is the last vertex of $\tau_i$ matched to $w_rw_{r+1}$ in the current enumeration.  
	We use the pair of cells $\mathcal{A}[r]$ and $\mathcal{B}[s]$ as the surrogate of the edge $v_{i,b}v_{i,b+1}$.   So we require $\mathcal{A}[r]$ and $\mathcal{B}[s]$ to be near $c_{r,2}$ and $c_{s,1}$, respectively, because $(c_{r,1},c_{r,2})$ is the surrogate of $w_rw_{r+1}$, and $(c_{s,1},c_{s,2})$ is the surrogate of $w_sw_{s+1}$.  We have to try all possible locations of $\mathcal{A}[r]$ and $\mathcal{B}[s]$ in the vicinity of $c_{r,2}$ and $c_{s,1}$.  $\mathcal{A}[r]$ and $\mathcal{B}[s]$ can be the surrogate for edges of multiple curves in $T$, which allows us to compare $\sigma$ with multiple input curves simultaneously at query time.  The constraint to be enforced is that $w_{r+1},\ldots,w_s$ can be matched to a segment joining a vertex $x_r$ of $\mathcal{A}[r]$ and a vertex $x_s$ of $\mathcal{B}[s]$ so that $d_F(x_r'x_s',\sigma[w_{r+1},w_s]) \leq (1+O(\eps))\delta$ for some subsegment $x_r'x_s' \subseteq x_rx_s$. Figure~\ref{fg:encoding} shows a illustration for the intuition above.
	
	\begin{figure}[h]
		\centerline{\includegraphics[scale=0.7]{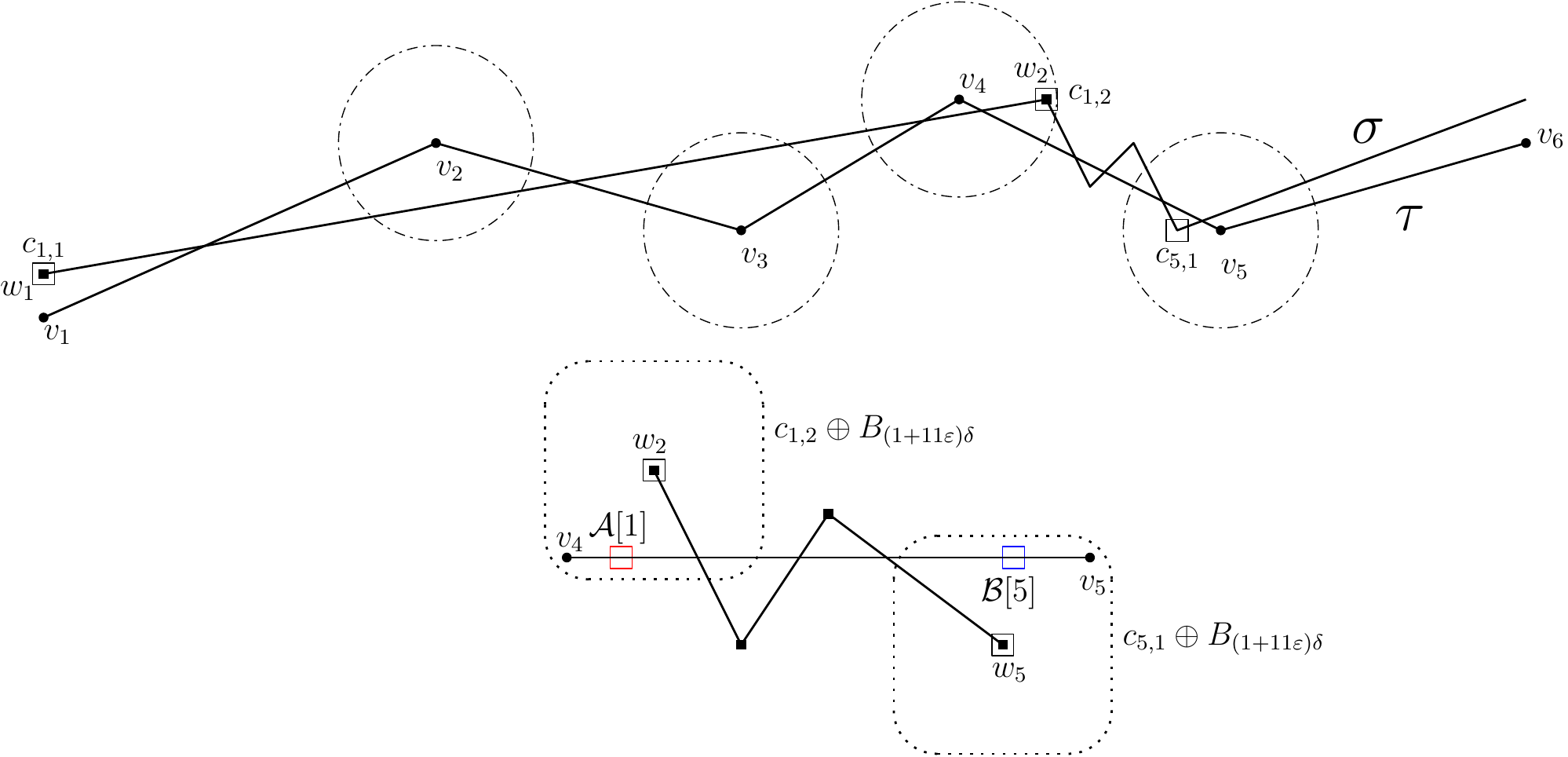}}
		\caption{The underlying intuition for deriving the coarse encoding of $\sigma$. We use $\tau$ instead of $\tau_i$ for ease of notation. Assume that $d_F(\tau, \sigma)\le \delta$ and the vertices $v_1, v_2,v_3, v_4$ are matched to the segment $w_1w_{2}$ by the Fr\'echet matching. $w_1w_2$ must intersect some balls centered at $\tau$'s vertices, which means that $w_1w_2$ intersects $\mathcal{G}_1$. Let $c_{1,1}$ and $c_{1,2}$ be the first and the last cells in $\mathcal{G}_1$ that intersect $w_1w_2$ along the direction of $w_1w_2$. $(c_{1,1}, c_{1,2})$ can serve as a surrogate of the edge $w_1w_2$ in a sense that we can verify whether $v_1, v_2,v_3, v_4$ can be matched to $w_1w_2$ properly by checking whether they can be matched to a segment that joins vertices of $c_{1,1}$ and $c_{1,2}$ properly. This idea can be generalized to all edges of $\sigma$ with vertices of $\tau$ matched to them. The subcurve $\sigma[w_2, w_5]$ is matched to an edge $v_4v_5$ of $\tau$. We introduce $\mathcal{A}[1] \subset G(c_{1,2}\oplus B_{(1+11\eps)\delta})$ and  $\mathcal{B}[5]\subset G(c_{5,1}\oplus B_{(1+11\eps)\delta})$. $(\mathcal{A}[1], \mathcal{B}[5])$ can serve as a surrogate of $v_4v_5$. $(\mathcal{A}[1], \mathcal{B}[5])$ can encode $\sigma[w_2,  w_5]$ sufficiently because for every segment $x_1x_5$ that joins a vertex $x_1$ of $\mathcal{A}[1]$ and a vertex $x_5$ of $\mathcal{B}[5]$, there exists a subsegment $x_1'x_5'\subset x_1x_5$ such that $d_F(x_1'x_5', \sigma[w_2, w_5])\le (1+O(\eps))\delta$.}
		\label{fg:encoding}
	\end{figure}
	
	We present the constraints for $(\mathcal{A},\mathcal{B},\mathcal{C})$ that realize the intuition above.  
	 When $(c_{j,1},c_{j,2}) \not= \text{null}$, a natural choice of $c_{j,1}$ is the first grid cell in $\mathcal{G}_1$ that we hit when walking from $w_j$ to $w_{j+1}$, i.e., segment shooting.  In $\real^d$ where $d \in \{2,3\}$, there are ray shooting data structures for boxes~\cite{BHO1994}.
	 In higher dimensions, ray shooting results are known for a single convex polytope and an arrangement of hyperplanes~\cite{AM1993}; even in such cases, the query time is substantially sublinear only if the space complexity is at least the input size raised to a power of $\Omega(d)$.  It would be $(mn)^{\Omega(d)}$ in our case.  We define a $\lambda$-segment  query problem below that approximates the ray shooting problem, and we will present an efficient solution for $\lambda = 11\eps\delta$ in Section~\ref{Sec: subroutines} that avoids an $(mn)^{\Omega(d)}$ term in the space complexity.  As mentioned before, the ray shooting result  in~\cite{BHO1994} suffices in two and three dimensions.
	\begin{quote}
		{\bf $\pmb{\lambda}$-segment query.}  A set $O$ of objects in $\real^d$ is preprocessed into a data structure so that for any oriented query segment $pq$, the \emph{$\lambda$-segment query with $pq$ on $O$} returns one of the following answers:
		\begin{itemize}
			\item If $pq$ intersects an object in $O$, let $x$ be the first intersection point with an object in $O$ as we walk from $p$ to $q$.  In this case, the query returns an object $o \in O$ such that $px$ intersects $o \oplus B_\lambda$.  Figure~\ref{fg:segment} shows an illustration.
			
			
			\item Otherwise, the query returns null or an object $o \in O$ such that $d(o,pq) \leq \lambda$.
		\end{itemize}
	\end{quote}	

	\begin{figure}
	\centerline{\includegraphics[scale=0.5]{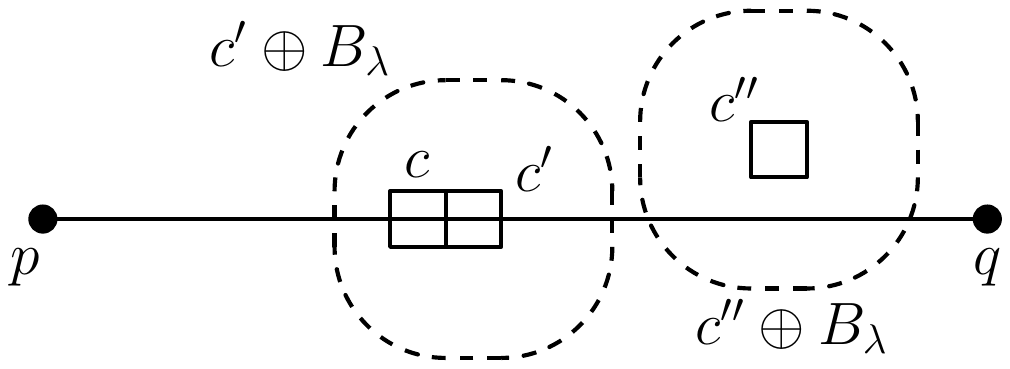}}
	\caption{The $\lambda$-segment query with $pq$ on the boxes $\{c,c',c''\}$ can return $c$ or $c'$ but not $c''$.}
	\label{fg:segment}
\end{figure}

\noindent We are now ready to state the three constraints on $(\mathcal{A},\mathcal{B},\mathcal{C})$.
	\begin{itemize}
		
		\item \textbf{Constraint~1:} 
		\begin{itemize}
			\item[(a)] Both $(c_{1,1}, c_{1,2})$ and $(c_{k-1,1}, c_{k-1,2})$ belong to $\mathcal{G}_1 \times \mathcal{G}_1$. 
			\item[(b)] For $j \in [k-1]$, if $(c_{j,1},c_{j,2}) \not= \text{null}$, then:
			\begin{itemize}
				\item[(i)] $c_{j,1}$ and $c_{j,2}$ are the grid cells returned by the $(11\eps\delta)$-segment queries with $w_jw_{j+1}$ and $w_{j+1}w_j$ on $\mathcal{G}_1$, respectively;
				\item[(ii)] the minimum point in $w_jw_{j+1} \cap (c_{j,1} \oplus B_{11\eps\delta})$ lies in front of the maximum point in $w_jw_{j+1} \cap (c_{j,2} \oplus B_{11\eps\delta})$ with respect to $\leq_{w_jw_{j+1}}$.
			\end{itemize}
		\end{itemize}
	
		\item {\bf Constraint~2:} 
		\begin{itemize}
			\item[(a)] $\mathcal{B}[1]$ and $\mathcal{A}[k-1]$ belong to $\mathcal{G}_2$.  
			\item[(b)] $w_1\in \mathcal{B}[1]$ and $w_k\in \mathcal{A}[k-1]$.
		\end{itemize}
		
		\item {\bf Constraint~3:} 
		\begin{itemize}
			\item[(a)] For $j \in [2,k-2]$, if $(c_{j,1},c_{j,2}) = \text{null}$, then $\mathcal{A}[j]$ and $\mathcal{B}[j]$ are null.
			\item[(b)] For $j \in [k-1]$, if $(c_{j,1}, c_{j,2}) \not= \text{null}$, then $\mathcal{A}[j]$ and $\mathcal{B}[j]$ belong to $\mathcal{G}_2$, $d(c_{j,1}, \mathcal{B}[j])\le (1+11\eps)\delta$, and $d(c_{j,2}, \mathcal{A}[j])\le (1+11\eps)\delta$.  
			\item[(c)] Let $\mathcal{J}$ be the set of $(r,s) \in [k-1]\times[k-1]$ such that $r < s$, $(c_{r,1}, c_{r, 2}) \not= \text{null}$, $(c_{s,1}, c_{s, 2}) \not= \text{null}$, and $(c_{j,1}, c_{j,2}) = \text{null}$ for $j \in [r+1,s-1]$.  For every $(r,s)\in \mathcal{J}$, let $x_r$ and $x_s$ be the smallest vertices of $\mathcal{A}[r]$ and $\mathcal{B}[s]$ according to the lexicographical order of their coordinates, there exist $x_r', x'_s \in x_rx_s$ such that $x'_r \leq_{x_rx_s} x'_s$ and $d_F(x_r'x_s', \sigma[w_{r+1}, w_{s}])\le (1+\eps)\delta$.  
		\end{itemize}
	
	\end{itemize}

	
   We remark that if $w_jw_{j+1}$ intersects the interior of the union of cells in $\mathcal{G}_1$, constraint~1(b)(ii) is satisfied automatically for $(c_{j,1},c_{j,2})$ given constraint~1(b)(i).  When $w_jw_{j+1}$ does not intersect the interior of the union of cells in $\mathcal{G}_1$, it is possible that the $(11\eps\delta)$-segment queries return two cells that violate constraint~1(b)(ii).   In this case, the input vertices are too far from $w_jw_{j+1}$ to be matched to any point in $w_jw_{j+1}$ within a distance $\delta$, so we can set $(c_{j,1},c_{j,2})$ to be null.

	The next result shows that any query curve $\sigma$ near a curve $\tau_i \in T$ has a coarse encoding with some additional properties.  These properties will be useful in the analysis.  
	Let $\mathcal{M}$ denote a matching between $\sigma$ and some $\tau_i \in T$.  For any subcurve $\sigma' \subseteq \sigma$, we use $\mathcal{M}(\sigma')$ to denote the subcurve of $\tau_i$ matched to $\sigma'$ by $\mathcal{M}$.
	
	\begin{lemma}\label{lem: existence_fuzzy}
		Let $\sigma = (w_1,\ldots,w_k)$ be a curve of $k$ vertices.  Let $\mathcal{M}$ be a matching between~$\sigma$ and $\tau_i \in T$ such that $d_{\mathcal{M}}(\tau_i,\sigma) \leq \delta$.    
		Let $\tilde{\pi}_j = \{v_{i,a} : a \in [m-1], v_{i,a} \in \mathcal{M}(w_jw_{j+1}) \setminus \mathcal{M}(w_j) \}$ for all $j \in [k-1]$.   Define $\pi_j = \tilde{\pi}_j$ for all $j \in [k-2]$, $\pi_{k-1} = \{v_{i,m}\} \cup \tilde{\pi}_{k-1}$, and $\pi_0 = \{v_{i,1}\ldots,v_{i,m}\} \setminus \bigcup_{j=1}^{k-1}\pi_j$.  There is a coarse encoding $(\mathcal{A},\mathcal{B},\mathcal{C})$ for $\sigma$ that satisfies the following properties. 
		\begin{enumerate}[(i)]
			\item For $j \in [2,k-1]$, $\pi_j = \emptyset$ if and only if $(c_{j,1},c_{j,2}) = \text{null}$.
			\item For all $(r,s) \in \mathcal{J}$, if $r = 1$ and $\pi_1 = \emptyset$, let $b_1 = 1$; otherwise, let $b_r = \max\{b : v_{i,b} \in \pi_r\}$.  For all $(r,s) \in \mathcal{J}$, there exist a point $z \in  \mathcal{A}[r] \cap \tau_{i,b_r}$ and another point $z' \in \mathcal{B}[s] \cap \tau_{i,b_r}$ such that $z \leq_{\tau_{i,b_r}} z'$.
		\end{enumerate}
	\end{lemma}
	
	\begin{proof}
		We define the component $\mathcal{C}$ as follows.  Given that $v_{i,1} \in \mathcal{M}(w_1)$, $w_1w_2$ intersects the interior of the union of cells in $\mathcal{G}_1$, so the $(11\eps\delta)$-segment query with $w_1w_2$ on $\mathcal{G}_1$ must return some cell; we define it to be $c_{1,1}$.  Similarly,  the $(11\eps\delta)$-segment query with $w_2w_1$ on $\mathcal{G}_1$ must return some cell; we define it to be $c_{1,2}$.  The pair $(c_{k-1,1},c_{k-1,2})$ are also defined in a similar way as $v_{i,m} \in \mathcal{M}(w_k)$.   Consider any $j \in [2,k-1]$.  If 
		$v_{i,a} \in \pi_j$ for some $a \in [m]$, then $v_{i,a} \in \mathcal{M}(w_jw_{j+1})$,
		which implies that $w_{j}w_{j+1}$ intersects $v_{i,a} \oplus B_\delta$ and hence the interior of the union of cells in $\mathcal{G}_1$.  Thus, $(c_{j,1},c_{j,2})$ can be defined using the $(11\eps\delta)$-segment queries with $w_jw_{j+1}$ and $w_{j+1}w_j$ as before. On the other hand, if $\pi_j = \emptyset$, we define $(c_{j,1}, c_{j,2})$ to be null.  Constraint 1 and property (i) in the lemma are thus satisfied.
		
		Next, we define $\mathcal{A}$ and $\mathcal{B}$ to satisfy constraints~2 and~3.
		
		As $v_{i,1}\in \mathcal{M}(w_1)$ and $v_{i,m}\in \mathcal{M}(w_k)$,  both $d(w_1,v_{i,1})$ and $d(w_k,v_{i,m})$ are at most $\delta$.  So $w_1$ lies in a cell in $G(v_{i,1} \oplus B_\delta) \subset G(v_{i,1}\oplus B_{(2+12\eps)\delta}) \subset \mathcal{G}_2$; we make this cell $\mathcal{B}[1]$.  Similarly, we define $\mathcal{A}[k-1]$ to be the cell in $\mathcal{G}_2$ that contains $w_k$.  Constraint~2 is thus enforced.  
		
		For $j \in [2,k-2]$, if $(c_{j,1},c_{j,2}) = \text{null}$, let $\mathcal{A}[j]$ and $\mathcal{B}[j]$ be null, satisfying constraint~3(a).  $\mathcal{B}[1]$ and $\mathcal{A}[k-1]$ have already been defined, and they belong to $\mathcal{G}_2$.  Since $w_1$ lies in a cell in $G(v_{i,1} \oplus B_\delta) \subset \mathcal{G}_1$, we have $w_1 \in c_{1,1} \oplus B_{11\eps\delta}$ by the $(11\eps\delta)$-segment query.  Then, $d(c_{1,1},\mathcal{B}[1]) \leq 11\eps\delta$ as $w_1 \in \mathcal{B}[1]$.  Similarly, $d(c_{k-1,2},\mathcal{A}[k-1]) \leq 11\eps\delta$.  So $\mathcal{B}[1]$ and $\mathcal{A}[k-1]$ satisfy constraint~3(b).  It remains to discuss $\mathcal{A}[j]$ for $j \in [1,k-2]$ and $\mathcal{B}[j]$ for $j \in [2,k-1]$.
		
		Consider an arbitrary $j_* \in [k-1]$ such that $(c_{j_*,1},c_{j_*,2}) \not= \text{null}$.  Recall that $\mathcal{J}$ is the set of $(r,s) \in [k-1]\times [k-1]$ such that $r < s$, $(c_{r,1}, c_{r, 2}) \not= \text{null}$, $(c_{s,1}, c_{s, 2}) \not= \text{null}$, and $(c_{j,1}, c_{j,2}) = \text{null}$ for $j \in [r+1,s-1]$.   Thus, if $j_* \leq k-2$, it must exist as the first value in exactly one element of $\mathcal{J}$, and if $j_* \geq 2$, it must also exist as the second value in exactly another element of $\mathcal{J}$.  As a result, it suffices to define $\mathcal{A}[r]$ and $\mathcal{B}[s]$ for every $(r,s) \in \mathcal{J}$ and verify that constraints~3(b) and~3(c) are satisfied.
		
		 Take any $(r,s) \in \mathcal{J}$.  If $\pi_r \not= \emptyset$, it is legal to define $b_r = \max\{b : v_{i,b} \in \pi_r\}$.  If $\pi_r = \emptyset$, then $r=1$ because for any $r > 1$, $\pi_r \not= \text{null}$ by (i) as $(c_{r,1},c_{r,2}) \not= \text{null}$ by the definition of $\mathcal{J}$.  In the case that $r=1$ and $\pi_1= \emptyset$, $b_1$ is defined to be 1.  Therefore, $b_r$ is well defined for all $(r,s) \in \mathcal{J}$.  The definition of $b_r$ implies that $b_r = \max\{b : v_{i,b} \in \mathcal{M}(w_rw_{r+1})\}$.  Since $(c_{s,1},c_{s,2}) \not= \text{null}$ and $(c_{j,1},c_{j,2}) = \text{null}$ for $j \in [r+1,s-1]$, by (i), $\pi_s \not= \emptyset$ and $\pi_j = \emptyset$ for $j \in [r+1,s-1]$.  It follows that $v_{i,b_r+1} \in \pi_s$ which is a subset of $\mathcal{M}(w_sw_{s+1})$.
		Pick any point $p \in w_rw_{r+1}$ and any point $q \in w_sw_{s+1}$ such that $v_{i,b_r} \in \mathcal{M}(p)$ and $v_{i,b_r+1} \in \mathcal{M}(q)$.  
		
		We claim that $pw_{r+1} \cap (c_{r,2} \oplus B_{11\eps\delta})$ and $w_sq \cap (c_{s,1} \oplus B_{11\eps\delta})$ are non-empty.   Since $c_{r,2}$ is the cell in $\mathcal{G}_1$ returned by the $(11\eps\delta)$-segment query with $w_{r+1}w_r$, for any intersection point $x$ between $w_{r+1}w_r$ and any cell in $\mathcal{G}_1$, we have $xw_{r+1} \cap (c_{r,2} \oplus B_{11\eps\delta}) \not= \emptyset$ by definition.  We have $p \in w_rw_{r+1} \cap (v_{i,b_r} \oplus B_\delta)$ by our choice of $p$; it means that $p$ is an intersection point between $w_rw_{r+1}$ and a cell in $G(v_{i,b_r} \oplus B_\delta) \subset \mathcal{G}_1$.   We can thus substitute $p$ for $x$ and conclude that $pw_{r+1} \cap (c_{r,2} \oplus B_{11\eps\delta}) \not= \emptyset$.   Similarly, we get $w_sq \cap (c_{s,1} \oplus B_{11\eps\delta}) \not= \emptyset$.
		
		By our claim, when we walk from $w_{r+1}$ to $p$, we hit $c_{r,2} \oplus B_{11\eps\delta}$ at some point $p'$, and when we walk from $w_s$ to $q$, we hit $c_{s,1} \oplus B_{11\eps\delta}$ at some point $q'$.  Pick two points $z_r \in \mathcal{M}(p')$ and $z_s \in \mathcal{M}(q')$.  By definition, $c_{r,2} \in G(v_{i_r,a_r}\oplus B_\delta)$ for some $\tau_{i_r} \in T$ and some index $a_r \in [m]$.  
		The cell width of $c_{r,2}$ is
		$\eps\delta/\sqrt{d}$, so $c_{r,2} \subset v_{i_r,a_r} \oplus B_{(1+\eps)\delta}$.
		By triangle inequality, $p' \in v_{i_r,a_r} \oplus B_{(1+12\eps)\delta}$ and hence $z_r \in v_{i_r,a_r} \oplus B_{(2+12\eps)\delta}$, which implies that $z_r$ is contained in a cell in $G(v_{i_r,a_r} \oplus B_{(2+12\eps)\delta}) \subset \mathcal{G}_2$.  By a similar reasoning, we can also deduce that $z_s$ is contained in a cell in 
		$\mathcal{G}_2$.  We define $\mathcal{A}[r]$ and $\mathcal{B}[s]$ to be the cells in $\mathcal{G}_2$ that contain $z_r$ and $z_s$, respectively.  Figure~\ref{fg:walk} shows an illustration.
		
		\begin{figure}
			\centerline{\includegraphics[scale=0.5]{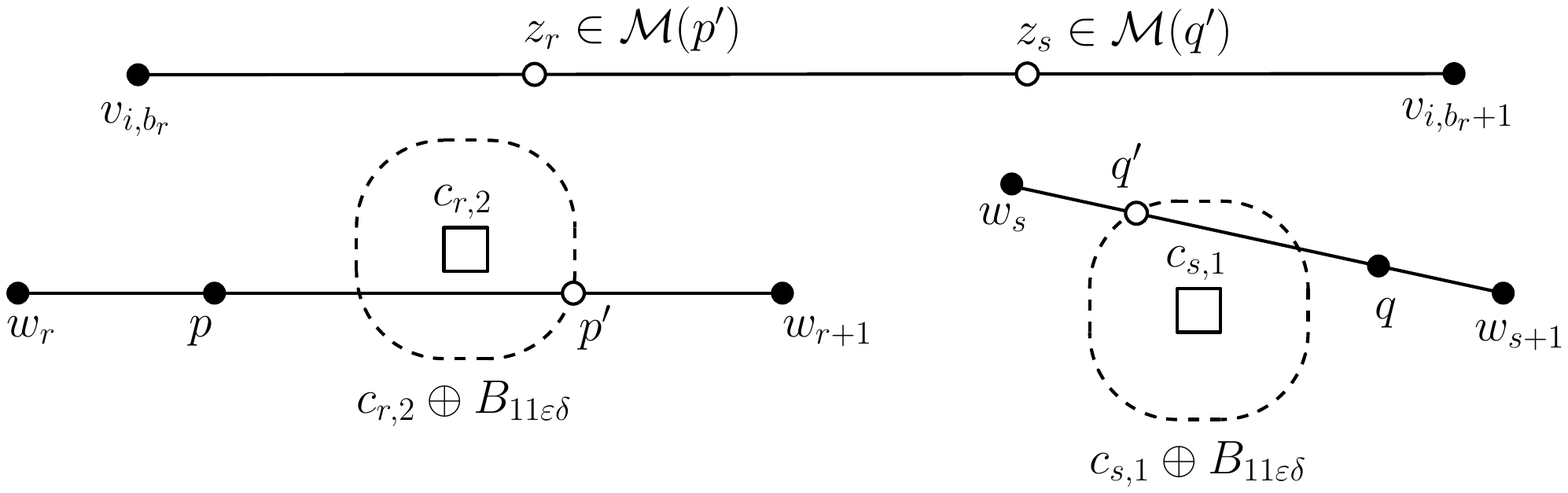}}
			\caption{Illustration of $p$, $p'$, $q'$, $q$, $z_r$, and $z_s$.}
			\label{fg:walk}
		\end{figure}
		
		Since $d(c_{r,2},z_r) \leq d(c_{r,2},p') + d(p',z_r) \leq (1+11\eps)\delta$, we get $d(c_{r,2},\mathcal{A}[r]) \leq (1+11\eps)\delta$.  Similarly, $d(c_{s,1},\mathcal{B}[s]) \leq (1+11\eps)\delta$.  This takes care of constraint~3(b).
		
		
		As $v_{i,b_r} \in \mathcal{M}(p)$ and $p \leq_{w_rw_{r+1}} p'$, we have $v_{i,b_r} \leq_{\tau_i} \mathcal{M}(p')$.
		Similarly, we have $\mathcal{M}(q') \leq_{\tau_i} v_{i,b_{r+1}}$.  Therefore, $v_{i,b_r} \leq_{\tau_i} \mathcal{M}(p') \leq_{\tau_i} \mathcal{M}(q') \leq_{\tau_i} v_{i,b_{r+1}}$.  As $z_r \in \mathcal{M}(p')$ and $z_s \in \mathcal{M}(q')$, $z_r$ and $z_s$ satisfy property (ii) of the lemma.
		The distance between $z_r$ and any vertex $x_r$ of $\mathcal{A}[r]$ is at most $\eps\delta$.  So is the distance between $z_s$ and any vertex $x_s$ of $\mathcal{B}[s]$.  Thus, we can use the linear interpolation $\mathcal{I}$ between $x_rx_s$ and $z_rz_s$ as a matching to get $d_\mathcal{I}(x_rx_s,z_rz_s) \leq \eps\delta$.  Combining $\mathcal{M}$ and $\mathcal{I}$ shows that there is a matching between $\sigma[p',q']$ and $x_rx_s$ within a distance of $(1+\eps)\delta$.  Since $\sigma[w_{r+1},w_s] \subseteq \sigma[p',q']$, we have thus verified constraint~3(c).
	\end{proof}
	
	

	\subsection{Data structure organization and construction}
	\label{sec:organize}
	
	We construct $\mathcal{G}_1$ and $\mathcal{G}_2$ in $O(mn/\eps^d)$ time and space.    We need a point location data structure for $\mathcal{G}_2$ which is organized as a multi-level tree as follows.  The top-level tree has leaves corresponding to the intervals of the cells on the first coordinate axis.  Each leaf is associated with the cells that project to the interval of that leaf, and these cells are stored in a second-level tree with leaves corresponding to the intervals of these cells on the second coordinate axis.  Continuing in this manner yields $d = O(1)$ levels, using $O(|\mathcal{G}_2|) = O(mn/\eps^d)$ space and $O\bigl((mn/\eps^d)\log\frac{mn}{\eps}\bigr)$ preprocessing time.  A point location takes $O(\log\frac{mn}{\eps})$ time.  
	
	The $(1+O(\eps),\delta)$-ANN data structure is a trie $\mathcal{D}$.  Each key to be stored in $\mathcal{D}$ is a \emph{candidate} coarse encoding, which is a 3-tuple $(\mathcal{A},\mathcal{B},\mathcal{C})$ just like a coarse encoding.  For a candidate coarse encoding, constraints~1(a), 2(a), 3(a), and 3(b) must be satisfied, but constraints~1(b), 2(b), and 3(c) are ignored.  This difference is necessary because constraints~1(b), 2(b), and 3(c) require the query curve, which is not available in preprocessing.  For each candidate coarse encoding $E$, let $T_E$ be the subset of input curves that are within a Fr\'{e}chet distance of $(1+O(\eps))\delta$ from any query curve that has $E$ as a coarse encoding, we will discuss shortly how to obtain the curves in $T_E$. 
	
	Each key $E$ in $\mathcal{D}$ has $O(k)$ size because $E$ stores $O(k)$ cells in $\mathcal{G}_1$ and $\mathcal{G}_2$.
	As a trie, $\mathcal{D}$ is a rooted tree with as many levels as the length of the key $E$.  Searching in $D$ boils down to visiting the appropriate child of an internal node of $\mathcal{D}$.  Each component of the key $E$ is an element of $\mathcal{G}_2 \cup \{\text{null}\}$ or $(\mathcal{G}_1 \times \mathcal{G}_1) \cup \{\text{null}\}$; there are  $O(m^2n^2/\eps^{2d})$ possibilities.  We keep a dictionary at each internal node of $\mathcal{D}$ for finding the appropriate child to visit in $O(\log \frac{mn}{\eps})$ time.  Hence, the total search time of $\mathcal{D}$ is $O(k\log \frac{mn}{\eps})$.  
	
	To bound the size of $\mathcal{D}$, observe that each key $E$ at a leaf of $\mathcal{D}$ induces $O(k)$ entries in the dictionaries at the ancestors of that leaf.  There are $O(\sqrt{d}/\eps)^{4d(k-1)}(mn)^{4(k-1)}$ candidate coarse encodings.   So the total space taken by the dictionaries at the internal nodes is $O(\sqrt{d}/\eps)^{4d(k-1)}(mn)^{4(k-1)} k$.  We will show that if a query curve has $E$ as a coarse encoding, any curve in $T_E$ is within a Fr\'{e}chet distance of $(1+O(\eps))\delta$ from that query curve.  Therefore, we only need to store one of the curves in $T_E$ at the leaf for $E$, and it suffices to store the index of this curve.  Therefore, the total space complexity of $\mathcal{D}$ is $O(\sqrt{d}/\eps)^{4d(k-1)}(mn)^{4(k-1)}k$.
	
	The construction of $\mathcal{D}$ proceeds as follows.  We initialize $\mathcal{D}$ to be empty.  We enumerate all possible candidate coarse encodings based on constraints~1(a), 2(a), 3(a), and 3(b).  Take a possible candidate coarse encoding $E$.  The set $T_E$ is initialized to be empty.  We go through every input curve $\tau_i$ to determine whether to include $\tau_i$ in $T_E$.  If $T_E \not= \emptyset$ in the end, we insert $E$ together with one curve in $T_E$ into $\mathcal{D}$.  In the following, we discuss the checking of whether to include $\tau_i$ in $T_E$.
	
	Let $E$ be $(\mathcal{A},\mathcal{B},\mathcal{C})$.   We generate all possible \emph{partitions} of $\{v_{i,1},\ldots,v_{i,m}\}$ that satisfy the following properties.
	\begin{quote}
		{\bf Partition:} a sequence of $k$ \emph{disjoint} subsets $(\pi_0,\pi_1,\ldots,\pi_{k-1})$ such that $\bigcup_{j = 0}^{k-1} \pi_j = \{v_{i,1},\ldots,v_{i,m}\}$, $v_{i,1} \in \pi_0$, $v_{i,m} \in \pi_{k-1}$, $\pi_j$ may be empty for some $j \in [k-2]$, and for any $v_{i,a} \in \pi_j$ and any $v_{i,b} \in \pi_l$, if $j < l$, then $a < b$.
	\end{quote}
	There are fewer than $m^{k-1}$ partitions.  Given a partition $(\pi_0,\ldots,\pi_{k-1})$, the vertices in $\pi_0$ are to be matched with $v_{i,1}$; for $j \in [k-1]$, the vertices in $\pi_j$ are to be matched with points in $w_jw_{j+1}\setminus \{w_j\}$, where $w_jw_{j+1}$ is the $j$-th edge of the query curve; $v_{i,m}$ and possibly other vertices of $\tau_i$ are matched with $w_k$.  The reference to $w_jw_{j+1}$ is conceptual; we do not need to know the query curve in preprocessing.
	
	We describe four tests for each partition below.  As soon as we come across a partition that passes all four tests, we insert $\tau_i$ into $T_E$.  If a partition fails any test, we move on to the next partition.  If no partition can pass all four tests in the end, we do not include $\tau_i$ in $T_E$.
	
	The first test is that for $j \in [2,k-1]$, $\pi_j = \emptyset$ if and only if $(c_{j,1},c_{j,2}) = \text{null}$.  This test takes $O(k)$ time.  We exclude $\pi_1$ from this test because $(c_{1,1},c_{1,2}) \not= \text{null}$ by constraint~1(a), whereas $\pi_1$ may be empty or not depending on the partition enumerated.
	
	In the second test, for $j \in [k-1]$, if $\pi_j \not= \emptyset$,  let $a_j, b_j \in [m]$ be the smallest and largest indices such that $v_{i,a_j},v_{i,b_j} \in \pi_j$, the intuition is that $v_{i,a_j},\ldots, v_{i,b_j}$ can be matched to the surrogate $(c_{j,1},c_{j,2})$ of $w_jw_{j+1}$ within a distance of $(1+O(\eps))\delta$.  The second test checks this property as follows.  Observe that $(c_{j,1},c_{j,2}) \not= \text{null}$: $(c_{1,1},c_{1,2}) \not= \text{null}$ by constraint~1(a), and for $j \in [2,k-1]$, $(c_{j,1},c_{j,2}) \not= \text{null}$ by the first test.  Pick the smallest vertices $x_j$ of $c_{j,1}$ and $y_j$ of $c_{j,2}$ according to the lexicographical order of their coordinates.  If $x_jy_j \cap (v_{i,a_j} \oplus B_{(1+12\eps)\delta})$ or $x_jy_j \cap (v_{i,b_j} \oplus B_{(1+12\eps)\delta})$ is empty, the test fails. Otherwise, compute the minimum point $x'_j$ in $x_jy_j \cap (v_{i,a_j} \oplus B_{(1+12\eps)\delta})$ and the maximum point $y'_j$ in $x_jy_j \cap (v_{i,b_j} \oplus B_{(1+12\eps)\delta})$ with respect to $\leq_{x_jy_j}$.  If it is not the case that $x'_j \leq_{x_jy_j} y'_j$, the test fails.  Suppose that $x'_j \leq_{x_jy_j} y'_j$.  Compute $d_F(x'_jy'_j,\tau_i[v_{i,a_j},v_{i,b_j}])$ and check whether it is $(1+12\eps)\delta$ or less.  If all of the above checks succeed for all $j \in [k-1]$, the second test succeeds; otherwise, the test fails.  The test takes $O(m\log m)$ time, which is dominated by the computation of $d_F(x'_jy'_j,\tau_i[v_{i,a_j},v_{i,b_j}])$ over all $j \in [k-1]$.
	
	The third test is that $\mathcal{B}[1] \in G(v_{i,1} \oplus B_\delta)$ and $\mathcal{A}[k-1] \in G(v_{i,m} \oplus B_\delta)$, which boils down to checking whether $d(v_{i,1},\mathcal{B}[1])$ and $d(v_{i,m},\mathcal{A}[k-1])$ are at most $\delta$. 

	The fourth test involves $\mathcal{J}$, the set of $(r,s) \in [k-1]\times[k-1]$ such that $r < s$, $(c_{r,1}, c_{r, 2}) \not= \text{null}$, $(c_{s,1}, c_{s, 2}) \not= \text{null}$, and $(c_{j,1}, c_{j,2}) = \text{null}$ for $j \in [r+1,s-1]$.   Note that $|\mathcal{J}| \leq k-1$ and it can be constructed in $O(k)$ time.  For every $(r,s) \in \mathcal{J}$, if $r = 1$ and $\pi_1 = \emptyset$, let $b_1 = 1$; otherwise, let $b_r = \max\{ b: v_{i,b} \in \pi_r\}$.  It follows that $b_r+1 = \min\{ a: v_{i,a} \in \pi_s\}$.  We check if it is the case that $\tau_{i,b_r} \cap \mathcal{A}[r] \not= \emptyset$, $\tau_{i,b_r} \cap \mathcal{B}[s] \not= \emptyset$, and we hit $\mathcal{A}[r]$ no later than $\mathcal{B}[s]$ when we walk from $v_{i,b_r}$ to $v_{i,b_r+1}$.  (Recall the intuition that the pair $\mathcal{A}[r]$ and $\mathcal{B}[s]$ serve as the surrogate of the edge $\tau_{i,b_r} = v_{i,b_r}v_{i,b_r+1}$.)  If check fails for any $(r,s) \in \mathcal{J}$, the test fails.  
	Otherwise, the test succeeds.   This test runs in $O(k)$ time.
	
	The following result summarizes the construction of $\mathcal{D}$ and four properties of each candidate coarse encoding in $\mathcal{D}$.
	
	\begin{lemma}
		\label{lem:D}
		The trie $\mathcal{D}$ has $O(\sqrt{d}/\eps)^{4d(k-1)}(mn)^{4(k-1)}k$ size and can be constructed in $O(\sqrt{d}/\eps)^{4d(k-1)}(mn)^{4(k-1)}(k\log \frac{mn}{\eps} + m^k\log m)$ time.  We can search $\mathcal{D}$ with a coarse encoding in $O(k\log \frac{mn}{\eps})$ time.  For each candidate coarse encoding $E = (\mathcal{A},\mathcal{B},\mathcal{C})$, a curve $\tau_i \in T$ belongs to $T_E$ if and only if there exists a partition $(\pi_0,\ldots,\pi_{k-1})$ of the vertices of $\tau_i$ that satisfy the following four properties.  For $j \in [k-1]$, if $j = 1$ and $\pi_1= \emptyset$, let $b_1 = 1$; otherwise, if $\pi_j \not= \emptyset$, let $a_j = \min\{a : v_{i,a} \in \pi_j\}$ and let $b_j = \max\{b : v_{i,b} \in \pi_j \}$.  
		\begin{enumerate}[(i)]
			\item For $j \in [2,k-1]$, $\pi_j = \emptyset$ if and only if $(c_{j,1},c_{j,2}) = \text{null}$.
			\item For $j \in [k-1]$, if $\pi_j \not= \emptyset$, let $x_j$ and $y_j$ be the smallest vertices of $c_{j,1}$ and $c_{j,2}$ according to the lexicographical order of their coordinates, there exist $x''_j, y''_j \in x_jy_j$ such that $x''_j \leq_{x_jy_j} y''_j$ and $d_F(x''_jy''_j,\tau_i[v_{i,a_j},v_{i,b_j}]) \leq (1+12\eps)\delta$.
			\item $\mathcal{B}[1] \in G(v_{i,1} \oplus B_\delta)$ and $\mathcal{A}[k-1] \in G(v_{i,m} \oplus B_\delta)$.
			\item For every $(r,s) \in \mathcal{J}$,  $\tau_{i,b_r} \cap \mathcal{A}[r] \not= \emptyset$, $\tau_{i,b_r} \cap \mathcal{B}[s] \not=\emptyset$, and we hit $\mathcal{A}[r]$ no later than $\mathcal{B}[s]$ when we walk from $v_{i,b_r}$ to $v_{i,b_r+1}$.
		\end{enumerate}
	\end{lemma}

	\subsection{Querying}
	\label{sec:ann-query}
	
	At query time, we are given a curve $\sigma = (w_1,..., w_k)$.   We enumerate all coarse encodings of $\sigma$; for each coarse encoding $E$ enumerated, we search the trie $\mathcal{D}$ for $E$; if $E$ is found, we return the curve in $T_E$ stored with $E$ as the answer of the query; if no coarse encoding of $\sigma$ can be found in $\mathcal{D}$, we return ``no''.  
	
	Each search in $\mathcal{D}$ takes $O(k\log\frac{mn}{\eps})$ time as stated in Lemma~\ref{lem:D}.  The enumeration of the coarse encodings of $\sigma$ require a solution for the $(11\eps\delta)$-segment queries on $\mathcal{G}_1$ as stated in constraint~1(b)(i) in Section~\ref{sec:coarse}.  We will discuss an efficient solution later.

	For $j \in [k-1]$, we make two $(11\eps\delta)$-segment queries with $w_jw_{j+1}$ and $w_{j+1}w_j$ on $\mathcal{G}_1$ to obtain $u_{j,1}$ and $u_{j,2}$, respectively.  If any of the two queries returns null, define $(u_{j,1},u_{j,2})$ to be null.  If $(u_{j,1},u_{j,2}) \not= \text{null}$ and the minimum point in $w_jw_{j+1} \cap (u_{j,1} \oplus B_{11\eps\delta})$ does not lie in front of the maximum point in $w_jw_{j+1} \cap (u_{j,2} \oplus B_{11\eps\delta})$ with respect to $\leq_{w_jw_{j+1}}$, then constraint~1(b)(ii) is  not satisfied.  It must  be the case that $w_jw_{j+1}$ does not intersect the interior of the union of cells in $\mathcal{G}_1$, and the $(11\eps\delta)$-segment queries just happen to return two cells that violate constraint~1(b)(ii).  In this case, the input vertices are too far from $w_jw_{j+1}$ to be matched to any point in $w_jw_{j+1}$ within a distance $\delta$, so we reset $(u_{j,1},u_{j,2})$ to be null.
	
	
	After defining $(u_{j,1}, u_{j,2})$ for $j\in [k-1]$, we generate the coarse encodings of $\sigma$ as follows.  The pairs $(c_{1,1}, c_{1,2})$ and $(c_{k-1,1},c_{k-1,2})$ are defined to be $(u_{1,1},u_{1,2})$ and $(u_{k-1,1},u_{k-1,2})$, respectively.  For $j \in [2,k-2]$,  we enumerate all possible $\mathcal{C}$ by setting $(c_{j,1}, c_{j,2})$ to be $(u_{j,1}, u_{j,2})$ or null.  This gives a total of $2^{k-3}$ possible $\mathcal{C}$'s.  We query the point location data structure for $\mathcal{G}_2$ to find the cells $\mathcal{B}[1]$ and $\mathcal{A}[k-1]$ that contain $w_1$ and $w_k$, respectively.  Then, for each $\mathcal{C}$ enumerated, we enumerate $\mathcal{A}[j]$ for $j \in [1,k-2]$ and $\mathcal{B}[j]$ for $j \in [2,k-1]$ according to constraints~3(a) and 3(b) in Section~\ref{sec:coarse}.  This enumeration produces $O(\sqrt{d}/\eps)^{2d(k-2)}$ tuples of  $(\mathcal{A},\mathcal{B},\mathcal{C})$.  For each $(\mathcal{A},\mathcal{B},\mathcal{C})$ enumerated, we check whether it satisfies constraint~3(c), which can be done in $O(k\log k)$ time as implied by the following result.
	
	\begin{lemma}
		\label{lem:check}
		Take any $(r,s) \in \mathcal{J}$.
		Let $x_r$ and $x_s$ be the smallest vertices of $\mathcal{A}[r]$ and $\mathcal{B}[s]$ by the lexicographical order of their coordinates.  We can check in $O((s-r)\log (s-r))$ time whether there are $x'_r,x'_s \in x_rx_s$ such that $x'_r \leq_{x_rx_s} x'_s$ and $d_F(x'_rx'_s,\sigma[w_{r+1},w_s]) \leq (1+\eps)\delta$.
	\end{lemma}
\cancel{
	\begin{proof}
		If $x_rx_s \cap (w_{r+1} \oplus B_{(1+\eps)\delta})$ or $x_rx_s \cap (w_s \oplus B_{(1+\eps)\delta})$ is empty, the required $x_r'$ and $x_s'$ do not exist.  Suppose not.
		Let $p$ be the minimum point in $x_rx_s \cap (w_{r+1} \oplus B_{(1+\eps)\delta})$ with respect to $\leq_{x_rx_s}$.  Let $q$ be the maximum point in $x_rx_s \cap (w_s \oplus B_{(1+\eps)\delta})$.   We can check $p \leq_{x_rx_s} q$ and compute $d_F(pq,\sigma[w_{r+1},w_s])$ in $O((s-r)\log(s-r))$ time.  We claim that there is a segment $x'_rx'_s \subseteq x_rx_s$ that satisfies the lemma if and only if $pq$ satisfies the lemma.  The reverse direction is trivial.  The forward direction can be proved in the same way as in the proof of Lemma~\ref{lem:D}.
	\end{proof}
}
	
	\begin{lemma}
		\label{lem:query-time}
		The query time is $O(kQ_{\mathrm{seg}})  + O(\sqrt{d}/\eps)^{2d(k-2)}k\log\frac{mn}{\eps}$, where $Q_\mathrm{seg}$ is the time to answer an $(11\eps\delta)$-segment query.
	\end{lemma}
\cancel{
	\begin{proof}
		We spend $O(\log\frac{mn}{\eps})$ time to obtain $\mathcal{B}[1]$ and $\mathcal{A}[k-1]$.  Then,
		we spend $O(kQ_\text{seg})$ time to obtain $(u_{j,1},u_{j,2})$ for $j \in [k-1]$ and $O(k\log\frac{mn}{\eps})$ search time for each coarse encoding of $\sigma$.   There are $2^{k-3}$ combinations in setting $(c_{j,1},c_{j,2})$ to be $(u_{j,1},u_{j,2})$ or null for $j \in [2,k-2]$.  This gives $2^{k-3}$ possible $\mathcal{C}$'s.  By constraint~3(b), for $j \in [1,k-2]$, we have $\mathcal{A}[j] \in G(c_{j,2} \oplus B_{(1+11\eps)\delta})$, and for $j \in [2,k-1]$, we have $\mathcal{B}[j] \in G(c_{j,1} \oplus B_{(1+11\eps)\delta})$.  Therefore, for each $\mathcal{C}$ enumerated, there are $O(\sqrt{d}/\eps)^{2d(k-2)}$ ways to set $\mathcal{A}$ and $\mathcal{B}$.  In all, the total number of $(\mathcal{A},\mathcal{B},\mathcal{C})$'s enumerated and checked is $O(\sqrt{d}/\eps)^{2d(k-2)}2^{k-3} = O(\sqrt{d}/\eps)^{2d(k-2)}$.  It takes $O(k\log k)$ time to check each by Lemma~\ref{lem:check}.
	\end{proof}
}

	\subsection{Approximation guarantee}
	
	
	First, we show that if $\sigma$ is within a Fr\'{e}chet distance $\delta$ from some input curve, there exists a coarse encoding $E$ of $\sigma$ such that $T_E \not= \emptyset$.  Hence, $E$ and a curve in $T_E$ are stored in $\mathcal{D}$.

	\begin{lemma}\label{lem: non_empty_entry}
		If $d_F(\tau_i, \sigma)\le \delta$, then $\tau_i \in T_E$ for some coarse encoding $E$ of $\sigma$.
	\end{lemma}
	
	\begin{proof}
		Let $\mathcal{M}$ be a Fr\'{e}chet matching between $\tau_i$ and $\sigma$.  Let $E$ be the coarse encoding specified for $\sigma$ in Lemma~\ref{lem: existence_fuzzy}.  For any subcurve $\sigma' \subseteq \sigma$, we use $\mathcal{M}(\sigma')$ to denote the subcurve of $\tau_i$ matched to $\sigma'$ by $\mathcal{M}$.  
		For $j \in [k-1]$, let $\tilde{\pi}_j = \{v_{i,a} : a \in [m-1], v_{i,a} \in \mathcal{M}(w_jw_{j+1}) \setminus \mathcal{M}(w_j) \}$.  Define $\pi_j = \tilde{\pi}_j$ for $j \in [k-2]$, $\pi_{k-1} = \{v_{i,m}\} \cup \tilde{\pi}_{k-1}$, and $\pi_0 = \{v_{i,1}, \ldots, v_{i,m}\} \setminus \bigcup_{j=1}^{k-1}\pi_j$.  We obtain a partition $(\pi_0,...,\pi_{k-1})$ of the vertices of $\tau_i$.   
		
		We prove that $E$, $\tau_i$, and $(\pi_0,..., \pi_{k-1})$ satisfy Lemma~\ref{lem:D}(i)--(iv) which put $\tau_i$ in $T_E$.  Lemma~\ref{lem:D}(i) follows directly from Lemma~\ref{lem: existence_fuzzy}(i), 
		
		Take any $j \in [k-1]$ such that $\pi_j \not= \emptyset$.  Let $\pi_{j}$ be $\{v_{i,a}, v_{i,a+1}, \ldots , v_{i,b}\}$.  By the definition of $\pi_j$, every vertex in $ \pi_{j}$ belongs to $\mathcal{M}(w_jw_{j+1})$, so $\tau_i[v_{i,a},v_{i,b}] \subset \mathcal{M}(w_jw_{j+1})$.  Then, there must exist two points $p,q \in w_jw_{j+1}$ such that $p\le_{w_jw_{j+1}} q$ and $d_F(pq, \tau_i[v_{i,a},v_{i,b}])\le \delta$.  If $j = 1$, we have $(c_{1,1},c_{1,2}) \not= \text{null}$ by constraint~1(a); if $j \in [2,k-1]$, by Lemma~\ref{lem: existence_fuzzy}(i), $(c_{j,1},c_{j,2}) \not= \text{null}$ as $\pi_j \not= \text{null}$.  Therefore, $c_{j,1}$ and $c_{j,2}$ are cells in $\mathcal{G}_1$ returned by the $(11\eps\delta)$-segment queries with $w_jw_{j+1}$ and $w_{j+1}w_j$, respectively.  We have shown that $d_F(pq, \tau_i[v_{i,a},v_{i,b}])\le \delta$; therefore, $p$ is contained in a cell in $G(v_{i,a}\oplus B_\delta) \subset \mathcal{G}_1$.  As $c_{j,1}$ is the cell returned by the $(11\eps\delta)$-segment query with $w_jw_{j+1}$, there must be a point $z_p \in w_jw_{j+1} \cap (c_{j,1} \oplus B_{11\eps\delta})$ such that $z_p \leq_{w_jw_{j+1}} p$.  In a similar way, we can conclude that there must be a point $z_q \in w_jw_{j+1} \cap (c_{j,2} \oplus B_{11\eps\delta})$ such that $q \leq_{w_jw_{j+1}} z_q$.  That is, $z_p \leq_{w_jw_{j+1}} p \leq_{w_jw_{j+1}} q \leq_{w_jw_{j+1}} z_q$.  Let $x_j$ and $y_j$ be the smallest vertices of $c_{j,1}$ and $c_{j,2}$ according to the lexicographical order of their coordinates.  Both $d(z_p,x_j)$ and $d(z_q,y_j)$ are at most $12\eps\delta$.  A linear interpolation from $z_pz_q$ to $x_jy_j$ maps $p$ and $q$ to two points $x''_j$ and $y''_j$ on $x_jy_j$, respectively, such that $x''_j \leq_{x_jy_j} y''_j$.  Also, the linear interpolation adds a distance $12\eps\delta$ or less, which gives $d_F(x''_jy''_j,\tau_i[v_{i,a},v_{i,b}]) \leq d_F(x''_jy''_j,pq) + d_F(pq,\tau_i[v_{i,a},v_{i,b}]) \leq (1+12\eps)\delta$.  Hence, Lemma~\ref{lem:D}(ii) is satisfied.
		
		The grid cells $\mathcal{B}[1]$ and $\mathcal{A}[k-1]$ are defined to contain $w_1$ and $w_k$, respectively.  Also, $v_{i,1}\in \mathcal{M}(w_1)$ and $v_{i,m}\in \mathcal{M}(w_k)$.  Hence, $\mathcal{B}[1]\in G(v_{i,1}\oplus B_{\delta})$ and $\mathcal{A}[k-1]\in G(v_{i,m}\oplus B_{\delta})$, satisfying Lemma~\ref{lem:D}(iii).
		
		For any pair $(r,s)\in \mathcal{J}$, by Lemma~\ref{lem: existence_fuzzy}(ii),  there exist two points $z \in \mathcal{A}[r] \cap \tau_{i,b_r}$ and $z' \in \mathcal{B}[s] \cap \tau_{i,b_r}$ such that $z \leq_{\tau_{i,b_r}} z'$.  Since $\mathcal{A}[r]$ and $\mathcal{B}[s]$ are interior-disjoint unless they are equal, we must hit $\mathcal{A}[r]$ no later than $\mathcal{B}[s]$ when we walk from $v_{i,b_r}$ to $v_{i,b_r+1}$, satisfying Lemma~\ref{lem:D}(iv).
	\end{proof}
	
	We show that if $E$ is a coarse encoding of $\sigma$, each curve in $T_E$ is close to $\sigma$.
	
	\begin{lemma}\label{lem: distance_bound}
		Let $E$ be a coarse encoding of  $\sigma$.  For every $\tau_i \in T_E$,  $d_F(\tau_i, \sigma)\le (1+24\eps)\delta$. 
	\end{lemma}
	\begin{proof} (Sketch)
		Suppose that $T_E \not= \emptyset$ as the lemma statement is vacuous otherwise.
		Take any $\tau_i \in T_E$.  We construct a matching $\mathcal{M}$ between $\tau_i$ and $\sigma$ such that $d_{\mathcal{M}}(\tau_i, \sigma)\le (1+24\eps)\delta$.  Since $T_E \not= \emptyset$, there exists a partition $(\pi_0,\ldots, \pi_{k-1})$ of the vertices of $\tau_i$ that satisfy Lemma~\ref{lem:D}(i)--(iv).  Using these properties, we can match the vertices of $\tau_i$ to points on $\sigma$ and then the vertices of $\sigma$ to points on $\tau_i$.  Afterwards, $\sigma$ and $\tau_i$ divided into line segments by their vertices and images of their matching partners.  We use linear interpolations to match the corresponding line segments.  More details can be found in the appendix.
	\end{proof}

	\cancel{

	Combining Lemmas~\ref{lem:D}, \ref{lem:query-time}, \ref{lem: non_empty_entry}, and ~\ref{lem: distance_bound}, we obtain the following result.
	
	\begin{lemma}\label{lem:ANN}
		Let $T$ be a set of $n$ polygonal curves in $\real^d$, each containing at most $m$ vertices.   Let $k \geq 3$ be the given maximum number of vertices in any query curve.  Let $P_\mathrm{seg}$, $S_\mathrm{seg}$, and $Q_\mathrm{seg}$ be the preprocessing time, space complexity, and query time for any data structure that solves the $(11\eps\delta)$-segment query on $\mathcal{G}_1$.   For any $\eps \in (0,1)$, there is a $(1+24\eps,\delta)$-ANN data structure for $T$ under the Fr\'{e}chet distance that uses $O(\sqrt{d}/\eps)^{4d(k-1)}(mn)^{4(k-1)}k + S_\mathrm{seg}$ space and $O(\sqrt{d}/\eps)^{4d(k-1)}(mn)^{4(k-1)}(k\log\frac{mn}{\eps} + m^k\log m) + P_\mathrm{seg}$ preprocessing time.  The query time is $O(kQ_{\mathrm{seg}}) + O(\sqrt{d}/\eps)^{2d(k-2)}k\log \frac{mn}{\eps}$.
	\end{lemma}

}

	In two and three dimensions, the ray shooting data structure for boxes in~\cite{BHO1994} can be used as the $(11\eps\delta)$-segment query data structure.  It has an $O(\log |\mathcal{G}_1|) = O(\log \frac{mn}{\eps})$ query time and an $O(|\mathcal{G}_1|^{2+\mu}) = O((mn)^{2+\mu}/\eps^{d(2+\mu)})$ space and preprocessing time for any fixed $\mu \in (0,1)$.   If the query segment does not intersect any cell in $\mathcal{G}_1$, we return null.  	In four and higher dimensions, we will prove the following result in Section~\ref{Sec: subroutines}.
	
	\begin{lemma}
		\label{lem:segment}
		We can construct a data structure in $O(\sqrt{d}/\eps)^{O(d/\eps^2)} \cdot \tilde{O}((mn)^{O(1/\eps^2)})$ space and preprocessing time such that given any oriented edge $e$ of the query curve $\sigma$, the data structure either discovers that $\min_{\tau_i\in T} d_F(\sigma,\tau_i) > \delta$, or reports a correct answer for the $(11\eps\delta)$-segment query with $e$ on $\mathcal{G}_1$.  The query time is $\tilde{O}((mn)^{0.5+\eps}/\eps^d)$.
	\end{lemma}
	
Combining the results in this section with the ray shooting result in~\cite{BHO1994}, Lemma~\ref{lem:segment}, and Theorem~\ref{thm:reduce} gives $(1+\eps)$-ANN data structures.  Theorem~\ref{thm:reduce} uses the deletion cost of a $(1+\eps,\delta)$-ANN data structure.  We perform each deletion by reconstructing the data structure from scratch because we do not have a more efficient solution.
	
	\begin{theorem}
		\label{thm:ann23}
		For any $\eps \in (0,0.5)$, there is a $(1+O(\eps))$-ANN data structure for $T$ under the Fr\'{e}chet distance with the following performance guarantees:
		\begin{itemize}
			\item $d \in \{2,3\}:$ \\
				query time $= O\bigl(\frac{1}{\eps}\bigr)^{2d(k-2)}k\log\frac{mn}{\eps}\log n$, \\
				space $= O\bigl(\frac{1}{\eps}\bigr)^{4d(k-1)+1}(mn)^{4(k-1)}k\log^2 n$, \\
				expected preprocessing time $= O\bigl(\frac{1}{\eps}\bigr)^{4d(k-1)}(mn)^{4(k-1)}\bigl(k\log\frac{mn}{\eps} + m^k\log m\bigr)n\log n$.
			\item $d \geq 4:$ \\
				query time $=\tilde{O}\bigl(\frac{1}{\eps^d}k(mn)^{0.5+\eps}\bigr) + O\bigl(\frac{\sqrt{d}}{\eps}\bigr)^{2d(k-2)}k\log \frac{mn}{\eps}\log n$, \\ 
				space $= O\bigl(\frac{\sqrt{d}}{\eps}\bigr)^{4d(k-1)}(mn)^{4(k-1)}k\cdot\frac{1}{\eps}\log^2 n + O\bigl(\frac{\sqrt{d}}{\eps}\bigr)^{O(d/\eps^2)} \cdot \tilde{O}((mn)^{O(1/\eps^2)})$, \\
				expected preprocessing time $= O\bigl(\frac{\sqrt{d}}{\eps}\bigr)^{4d(k-1)}(mn)^{4(k-1)}\bigl(k\log\frac{mn}{\eps} + m^k\log m\bigr)n\log n \,\, +$ \\
			\hspace*{132pt}$O\bigl(\frac{\sqrt{d}}{\eps}\bigr)^{O(d/\eps^2)} \cdot \tilde{O}((mn)^{O(1/\eps^2)})$.
		\end{itemize}
	\end{theorem}

	
	\section{$\pmb{(3+O(\eps),\delta)}$-ANN}
	\label{sec:3ANN}
	
	Given a query curve $\sigma = (w_1,w_2,\ldots,w_k)$, for $j \in [k-1]$, we solve the $(11\eps\delta)$-segment queries with $w_jw_{j+1}$ and $w_{j+1}w_j$ on $\mathcal{G}_1$ as before.  Let $((c_{j,1}, c_{j,2}))_{j \in [k-1]}$ denote the results of the queries.  Recall that each $(c_{j,1},c_{j,2})$ belongs to $(\mathcal{G}_1 \times \mathcal{G}_1) \cup \{\text{null}\}$.
	
	Suppose that there are $k_0 \leq k-1$ non-null pairs in $((c_{j,1}, c_{j,2}))_{j \in [k-1]}$.  Extract these non-null pairs to form the sequence $((c_{j_r,1}, c_{j_ r,2}))_{r \in [k_0]}$.  Note that $j_1 = 1$ and $j_{k_0} = k-1$ by constraint~1(a).  We construct a polygonal curve $\sigma_0$ by connecting the centers of $c_{j_r,1}$ and $c_{j_r,2}$ for $r \in [k_0]$ and the centers of $c_{j_r,2}$ and $c_{j_{r+1},1}$ for $r \in [k_0-1]$.  The polygonal curve $\sigma_0$ acts as a surrogate of $\sigma$.  It has at most $2k-2$ vertices.
	We will use $\sigma_0$ as the key to search a trie at query time to obtain an answer for a $(3+O(\eps),\delta)$-ANN query.  As a result, no enumeration is needed which avoids the exponential dependence of the query time on $k$.
	
	In preprocessing, we enumerate all sequences of $2l$ cells in $\mathcal{G}_1$ for $l \in [2,k-1]$.  For each sequence, we construct the polygonal curve $\sigma'$ that connects the centers of the cells in the sequence, and we find the nearest input curve $\tau_i$ to $\sigma'$.  If $d_F(\sigma',\tau_i) \leq (1+12\eps)\delta$, we store $(\sigma',i)$ in a trie $\mathcal{D}$.  There are $O(\sqrt{d}/\eps)^{2d(k-1)}(mn)^{2(k-1)}$ entries in $\mathcal{D}$.  We organize the trie $\mathcal{D}$ in the same way as described in Section~\ref{sec:organize}. The space required by $\mathcal{D}$ is $O(\sqrt{d}/\eps)^{2d(k-1)}(mn)^{2(k-1)}k$.  The search time of $\mathcal{D}$ is $O(k\log\frac{mn}{\eps})$.  The preprocessing time is  $O(\sqrt{d}/\eps)^{2d(k-1)}(mn)^{2(k-1)}(k\log\frac{mn}{\eps} + kmn\log(km)) = O(\sqrt{d}/\eps)^{2d(k-1)}(mn)^{2k-1}k\log\frac{mn}{\eps}$ due to the computation of the nearest input curve for each sequence enumerated.
	
	At query time, we construct $\sigma_0$ from $\sigma$ in $O(kQ_\mathrm{seg})$ time, where $Q_\mathrm{seg}$ is the time to answer a $(11\eps\delta)$-segment query.  We compute $d_F(\sigma,\sigma_0)$ in $O(k^2\log k)$ time.  If $d_F(\sigma,\sigma_0) > (2+12\eps)\delta$, we report ``no''.  Otherwise, we search $\mathcal{D}$ with $\sigma_0$ in $O(k\log\frac{mn}{\eps})$ time.   If the search fails, we report ``no''.  Otherwise, the search returns $(\sigma_0,i)$ for some $i \in [n]$.
	
	\cancel{
	\begin{lemma}
		\label{lem:D2}
		The query time is $O(kQ_\mathrm{seg} + k\log\frac{mn}{\eps})$.  The preprocessing time is $P_\mathrm{seg} + O(\sqrt{d}/\eps)^{2d(k-1)}(mn)^{2k-1}k\log\frac{mn}{\eps}$ time.  The space is $S_\mathrm{seg} + O(\sqrt{d}/\eps)^{2d(k-1)}(mn)^{2(k-1)}k$.  
	\end{lemma}
}
	
	
	\begin{lemma}
		\label{lem:query2}
		If  $d_F(\sigma,\sigma_0) \leq (2+12\eps)\delta$ and the search in $\mathcal{D}$ with $\sigma_0$ returns $(\sigma_0,i)$, then $d_F(\sigma,\tau_i) \leq (3+24\eps)\delta$.  Otherwise, $\min_{\tau_i \in T} d_F(\sigma,\tau_i) > \delta$.
	\end{lemma}

	Combining the results in this section with the ray shooting results in two and three dimensions~\cite{BHO1994}, Lemma~\ref{lem:segment}, and Theorem~\ref{thm:reduce}, we obtain the following theorem.
	
	\cancel{
	\begin{lemma}\label{lem:3ANN}
		Let $T$ be a set of $n$ polygonal curves in $\real^d$, each containing at most $m$ vertices.   Let $k \geq 3$ be the given maximum number of vertices in any query curve.  Let $P_\mathrm{seg}$, $S_\mathrm{seg}$, and $Q_\mathrm{seg}$ be the preprocessing time, space complexity, and query time for any data structure that solves the $(11\eps\delta)$-segment query on $\mathcal{G}_1$.   For any $\eps \in (0,1)$, there is a $(3+34\eps,\delta)$-ANN data structure for $T$ under the Fr\'{e}chet distance that uses $O(\sqrt{d}/\eps)^{2d(k-1)}(mn)^{2(k-1)}k + S_\mathrm{seg}$ space and $O(\sqrt{d}/\eps)^{2d(k-1)}(mn)^{2k-1}k\log\frac{mn}{\eps} + P_\mathrm{seg}$ preprocessing time.  The query time is $O(kQ_{\mathrm{seg}} + k\log\frac{mn}{\eps})$.
	\end{lemma}

	Substituting the ray shooting result~\cite{} for boxes in two and three dimensions yields the following theorem.
}
	
	\begin{theorem}
		\label{thm:3ANN}
		For any $\eps \in (0,0.5)$, there is a $(3+O(\eps))$-ANN data structure for $T$ under the Fr\'{e}chet distance with the following performance guarantees:
		\begin{itemize}
			\item $d\in\{2,3\}$: \\
						query time $= O(k\log\frac{mn}{\eps}\log n)$, \\
						space $= O\bigl(\frac{1}{\eps}\bigr)^{2d(k-1)+1}(mn)^{2(k-1)}k\log^2 n$, \\
						expected preprocessing time $= O\bigl(\frac{1}{\eps}\bigr)^{2d(k-1)}(mn)^{2k-1}kn\log\frac{mn}{\eps}\log n$.
			
			\item $d \geq 4:$ \\
						query time $= \tilde{O}\bigl(\frac{1}{\eps^d}k(mn)^{0.5+\eps}\bigr)$, \\
						space $= O\bigl(\frac{\sqrt{d}}{\eps}\bigr)^{2d(k-1)}(mn)^{2(k-1)}k\cdot\frac{1}{\eps}\log^2 n +  O\bigl(\frac{\sqrt{d}}{\eps}\bigr)^{O(d/\eps^2)} \cdot \tilde{O}((mn)^{O(1/\eps^2)})$, \\
						expected preprocessing time $= O\bigl(\frac{\sqrt{d}}{\eps}\bigr)^{2d(k-1)}(mn)^{2k-1}kn\log\frac{mn}{\eps}\log n \,\, +$ \\ \hspace*{132pt}$O\bigl(\frac{\sqrt{d}}{\eps}\bigr)^{O(d/\eps^2)} \cdot \tilde{O}((mn)^{O(1/\eps^2)})$.
		\end{itemize}
	\end{theorem}

	\section{$\pmb{(11\eps\delta)}$-segment queries and proof of Lemma~\ref{lem:segment}}\label{Sec: subroutines}
	
	We describe the $(11\eps\delta)$-segment query data structure in Lemma~\ref{lem:segment}.  We first present the main ideas before giving the details.  Let $w_jw_{j+1}$ be a query segment, which is unknown at preprocessing.  There are three building blocks.  
	
	First, the intuition is to capture the support lines of all possible query segments using pairs of cells in $\mathcal{G}_1$.  It would be ideal to retrieve a pair of cells intersected by $\aff(w_jw_{j+1})$, but this seems to be as difficult as the ray shooting problem.  For a technical reason, we need to use more grid cells in a larger neighborhood of the input vertices than in $\mathcal{G}_1$, so define $\mathcal{G}_3 = \bigcup_{i \in [n], a \in [m]} G(v_{i,a} \oplus B_{(1+6\eps)\delta})$.  
	
	We find a grid vertex $x$ of $\mathcal{G}_1$ that is a $(1+\eps)$-approximate nearest grid vertex to $\aff(w_jw_{j+1})$.  
	We will show that if $d(x,\aff(w_jw_{j+1})) > (1+\eps)\eps\delta$, the answer to the $(11\eps\delta)$-segment query is null; otherwise, we can find a cell $\gamma \in \mathcal{G}_3$ near $x$ that intersects $\aff(w_jw_{j+1})$.  We can use any other cell $c \in \mathcal{G}_1$ to form a pair with $\gamma$ that acts as a surrogate for the support lines of query segments that pass near $c$ and $\gamma$.
	
	Second, given $w_jw_{j+1}$ at query time, among all possible choices of $c$, we need to find the right one(s) efficiently so that $(c,\gamma)$ is a surrogate for $\aff(w_jw_{j+1})$.  We explain the ideas using the case that $w_{j+1}$ lies between $w_j$ and $\aff(w_jw_{j+1}) \cap \gamma$.   Note that $w_j$ may not be near any cell in $\mathcal{G}_1$.  In order that $\min_{\tau_i \in T} d(\sigma,\tau_i) \leq \delta$, $w_j$ must be within a distance $\delta$ from some input edge $\tau_{i,a}$.  We find a maximal packing of $\aff(\tau_{i,a}) \oplus B_{O(\delta)}$ using lines that are parallel to $\tau_{i,a}$ and are at distance $\Theta(\eps\delta)$ or more apart.  There are $O(\eps^{1-d})$ lines in the packing, and every point in $\aff(\tau_{i,a}) \oplus B_\delta$ is within a distance $O(\eps\delta)$ from some line in the packing.  The projection of $w_j$ to the approximately nearest line approximates the location of $w_j$.  Hence, we should seek to divide the lines in the  packing into appropriate segments so that, given $w_j$ and its approximately nearest line in the packing, we can efficiently find the segment that contains  the projection of $w_j$ and retrieve some precomputed information for that segment.
	
	Third, let $\ell$ be a line in the packing mentioned above, for each possible cell $c \in \mathcal{G}_1$, we use the geometric construct $F(c,\gamma)= \{x \in \real^d : \exists \,\, y \in \gamma \, \text{s.t.} \,\, xy \cap c \not= \emptyset\}$ defined in~\cite{cheng2022curve} which can be computed in $O(1)$ time.  The projection of $(\ell \oplus B_{2\eps\delta})\cap F(c,\gamma)$ in $\ell$ is the set of points on $\ell$ such that if the projection of $w_j$ is in it, then $(c,\gamma)$ is a surrogate for $\aff(w_jw_{j+1})$.   As a result, the endpoints of the projections of $(\ell \oplus B_{2\eps\delta}) \cap F(c,\gamma)$ over all possible choices of $c$ divide $\ell$ into segments that we desire.  Each segment may stand for several choices of $c$'s.  For each segment, we store the cell $c'$ close to that segment because the ideal choice is the cell that we hit first as we walk from $w_j$ to $w_{j+1}$.  
	
	As described above, we use two approximate nearest neighbor data structures that involve lines.  The first one is due to Andoni~et~al.~\cite{andoni2009approximate} which stores a set of points $P$ such that given a query line, the $(1+\eps)$-approximate nearest point to the query line can be returned in $\tilde{O}(d^3|P|^{0.5+\eps})$ time.  It uses $\tilde{O}\bigl(d^2|P|^{O(1/\eps^2)}\big)$ space and preprocessing time.  The second result is due to Agarwal~et~al.~\cite{agarwal2017approximate} which stores a set $L$ of lines such that given a query point, the 2-approximate nearest line to the query point can be returned in $\tilde{O}(1)$ time.  It uses $\tilde{O}(|L|^2)$ space and expected preprocessing time.

\cancel{
	\begin{theorem}[\cite{andoni2009approximate}]\label{thm:ANP}
	Let $P$ be a set of points in $\real^d$.  For any $\mu \in (0,0.5)$, there exists a data structure that uses $\tilde{O}\bigl(d^2|P|^{O(1/\mu^2)}\big)$ space and preprocessing time such that for any query line, the $(1+\mu)$-approximately nearest point in $P$ can be reported in $\tilde{O}(d^3|P|^{0.5+\mu})$ time.
	\end{theorem}
	The second result is to preprocess a set of lines such that given a query point, the $(1+\eps)$-approximate nearest line to the query point can be returned efficiently.  The following result is due to Agarwal~et~al.~\cite{agarwal2017approximate}.
	
		\begin{theorem}[\cite{agarwal2017approximate}]\label{thm:ANL}
		Let $L$ be a set of lines in $\real^d$.  For any $\mu \in (0,1)$, there exists a data structure that uses $\tilde{O}(|L|^2/\mu^{d-1})$ space and expected preprocessing time such that for any query point, the $(1+\mu)$-approximately nearest line in $L$ can be reported in $\tilde{O}(\mu^{1-d})$ time.
	\end{theorem}
}
	
	\subsection{Data structure organization}
	
	We restrict $\eps$ to be chosen from $(0,0.5)$.  We construct the data structure of Andoni~et~al.~\cite{andoni2009approximate} for the grid vertices of $\mathcal{G}_1$ so that for any query line, the $(1+\eps)$-approximate nearest grid vertex can be returned in $\tilde{O}((mn)^{0.5+\eps}/\eps^d)$ time.  We denote this data structure by $D_\mathrm{anp}$.  It takes $O(\sqrt{d}/\eps)^{O(d/\eps^2)} \cdot \tilde{O}((mn)^{O(1/\eps^2)})$ space and preprocessing time.
	
	For each input edge $\tau_{i,a}$, define a set of lines $L_{i,a}$ as follows.  Let $H$ be the hyperplane through $v_{i,a}$ orthogonal to $\aff(\tau_{i,a})$.  Take a $(d-1)$-dimensional grid in $H$ with $v_{i,a}$ as a grid vertex and cell width $\eps\delta/\sqrt{d-1}$.  The set $L_{i,a}$ includes every line that is orthogonal to $H$ and passes through a vertex of this grid in $H$ at distance within $(1+2\eps)\delta$ from $v_{i,a}$.  The set $L_{i,a}$ has $O(\eps^{1-d})$ size, and it can be constructed in $O(\eps^{1-d})$ time.  Moreover, every point in the cylinder $\aff(\tau_{i,a}) \oplus B_\delta$ is within a distance $\eps\delta$ from some line in $L_{i,a}$.  
	
	Define $\mathcal{L} = \bigcup_{i\in[n], a\in[m-1]}L_{i,a}$.  The size of $\mathcal{L}$ is $O(mn/\eps^{d-1})$, and $\mathcal{L}$ can be constructed in $O(mn/\eps^{d-1})$ time.  We construct the data structure of Agarwal~et~al.~\cite{agarwal2017approximate} for $\mathcal{L}$ so that for any query point, a 2-approximate nearest line in $\mathcal{L}$ can be returned in $\tilde{O}(1)$ time.  We denote this data structure by $D_\mathrm{anl}$.  It uses $\tilde{O}((mn)^2/\eps^{2d-2})$ space and expected preprocessing time.
	
	Recall that $\mathcal{G}_3 = \bigcup_{i \in [n], a \in [m]} G(v_{i,a} \oplus B_{(1+6\eps)\delta})$.
	
	For every $\gamma \in \mathcal{G}_3$ and every $c \in \mathcal{G}_1$, we construct $F(c,\gamma) = \{ x \in \real^d : \exists \, y \in \gamma \,\, \mbox{s.t.} \,\,  xy\cap c \not= \emptyset\}$, which is empty or an unbounded convex polytope of $O(1)$ size that can be constructed in $O(1)$ time as a Minkowski sum~\cite{cheng2022curve}.    The total time needed is $O((mn)^2/\eps^{2d})$.
	
	For every $\gamma \in \mathcal{G}_3$, every $c \in \mathcal{G}_1$, and every line $\ell \in \mathcal{L}$, compute the intersection $(\ell \oplus B_{2\eps\delta}) \cap F(c,\gamma)$ and project it orthogonally to a segment in $\ell$.  Take any line $\ell \in \mathcal{L}$.  The resulting segment endpoints in $\ell$ divide $\ell$ into \emph{canonical segments}.  There are $O((mn)^2/\eps^{2d})$ canonical segments in $\ell$.  For every cell $\gamma \in \mathcal{G}_3$ and every canonical segment $\xi \subseteq \ell$, compute the set $C_{\gamma,\xi}$ of every cell $c \in \mathcal{G}_1$ such that $\xi$ is contained in the projection of $(\ell \oplus B_{2\eps\delta}) \cap F(c,\gamma)$ onto $\ell$.  Fix an arbitrary point in $\xi$ and denote it by $p_\xi$.    Each $C_{\gamma,\xi}$ has $O(mn/\eps^{d})$ size.  The total time needed over all cells in $\mathcal{G}_3$ and all canonical segments in all lines in $\mathcal{L}$ is $\tilde{O}((mn)^5/\eps^{5d-1})$.
	
	Let $p_\gamma$ be the center of the cell $\gamma$.  Define $c_{\gamma,\xi}$ to be the cell in $C_{\gamma,\xi}$ such that $p_\xi p_\gamma\cap (c_{\gamma,\xi} \oplus B_{5\eps\delta})$ is nearest to $p_\xi$ among $\{p_\xi p_\gamma\cap (c \oplus B_{5\eps\delta}) : c \in C_{\gamma,\xi}\}$.
	 The total time to compute $c_{\gamma,\xi}$ over all cells in $\mathcal{G}_3$ and all canonical segments in all lines in $\mathcal{L}$ is $O((mn)^5/\eps^{5d-1})$.
	
	Finally, for every line $\ell \in \mathcal{L}$, we store the canonical segments in $\ell$ in an interval tree $T_\ell$~\cite{CLRS}.    It uses linear space and preprocessing time.  For any query point in $\ell$, one can search $T_\ell$ in $O(\log \frac{mn}{\eps})$ time to find the canonical segment in $\ell$ that contains the query point.  For each canonical segment $\xi$ stored in $T_\ell$, we keep a dictionary $T_{\xi}$ that stores the set $\{(\gamma,c_{\gamma,\xi}) : \gamma \in \mathcal{G}_3\}$ with $\gamma$ as the key.  For any cell $\gamma \in \mathcal{G}_3$, we can search $T_\xi$ in $O(\log \frac{mn}{\eps})$ time to report $c_{\gamma,\xi}$.   These interval trees and dictionaries have a total size of $O((mn)^4/\eps^{4d-1})$, and they can be constructed in $\tilde{O}((mn)^4/\eps^{4d-1})$ time.
	
	The data structures $D_\mathrm{anp}$, $D_\mathrm{anl}$, $T_\ell$ for $\ell \in \mathcal{L}$, and $T_\xi$ for all canonical segments $\xi$'s are what we need to support the $(11\eps\delta)$-segment queries on $\mathcal{G}_1$.  
	
	\begin{lemma}
		\label{lem:D-seg}
		We can construct $D_\mathrm{anp}$, $D_\mathrm{anl}$, $T_\ell$ for $\ell \in \mathcal{L}$, and $T_\xi$ for every $\ell \in \mathcal{L}$ and every canonical segment $\xi \subset \ell$ in $O(\sqrt{d}/\eps)^{O(d/\eps^2)} \cdot \tilde{O}((mn)^{O(1/\eps^2)})$ space and preprocessing time.
	\end{lemma}

	In the definition of $c_{\gamma,\xi}$, one may ask what if $p_\xi p_\gamma$ does not intersect $c \oplus B_{5\eps\delta}$ for some $c \in C_{\gamma,\xi}$.  We prove that this cannot happen.  We also establish some other properties.
	
	\begin{lemma}
		\label{lem:C}
		Let $\gamma$ be a cell in $\mathcal{G}_3$.  Let $\xi$ be a canonical segment.
		Let $L_\xi$ be the cylinder with $\xi$ as the axis and radius $2\eps\delta$.
		\begin{enumerate}[(i)]
			\item For every cell $c \in \mathcal{G}_1$, if $c \cap xy \not= \emptyset$ for some points $x \in L_\xi$ and $y \in \gamma$, then $c \in C_{\gamma,\xi}$.
			\item For every point $x \in L_\xi$, every point $y \in \gamma$ and every cell $c \in C_{\gamma,\xi}$, $xy \cap (c \oplus B_{5\eps\delta}) \not= \emptyset$.
			\item Let $\lambda$ be any value greater than or equal to $11\eps\delta$.  When we walk from a point $x \in L_\xi$ to a point $y \in \gamma$, we cannot hit any $c \in C_{\gamma,\xi}$ earlier than $c_{\gamma,\xi} \oplus B_\lambda$ irrespective of the choices of $x$ and $y$.
		\end{enumerate}
	\end{lemma}

	\subsection{Answering a query}
	\label{sec:segment-query}
	
	Given an oriented segment $w_jw_{j+1}$ of the query curve $\sigma$, we answer the $(11\eps\delta)$-segment query with $w_jw_{j+1}$ on $\mathcal{G}_1$ by the following steps.
	\begin{itemize}
		\item Step~1: We query $D_\mathrm{anp}$ with $\aff(w_jw_{j+1})$ to report a grid vertex $x$ of $\mathcal{G}_1$.   This takes $\tilde{O}((mn)^{0.5+\eps}/\eps^d)$ time.  
		
		\vspace{6pt}
		
		\item Step~2: We check the distance $d(x,\aff(w_jw_{j+1}))$.  If $d(x,\aff(w_jw_{j+1})) > (1+\eps)\eps\delta$, then $\aff(w_jw_{j+1})$ is at distance more than $\eps\delta$ from the closest grid vertex of $\mathcal{G}_1$, which implies that $\aff(w_jw_{j+1})$ does not intersect any cell in $\mathcal{G}_1$.  In this case, we report null.
		
		We also check the distances $d(w_j,\mathcal{L})$ and $d(w_{j+1},\mathcal{L})$.  We query $D_\mathrm{anl}$ with $w_j$ in $\tilde{O}(1)$ time to find a line $\ell_j \in \mathcal{L}$.  If $d(w_j,\ell_j) > 2\eps\delta$, then $d(w_j,\mathcal{L}) > \eps\delta$, which implies that $w_j$ is at distance farther than $\delta$ from $\aff(\tau_{i,a})$ for any $\tau_i \in T$ and any $a \in [m-1]$.  As a result, $d_F(\sigma,\tau_i) > \delta$ for all $\tau_i \in T$, so we report ``no'' for the $(\kappa,\delta)$-ANN query. Analogously, we query $D_\mathrm{anl}$ with $w_{j+1}$ in $\tilde{O}(1)$ time to find a line $\ell_{j+1} \in \mathcal{L}$.  If $d(w_{j+1},\ell_{j+1}) > 2\eps\delta$, we report ``no'' for the $(\kappa,\delta)$-ANN query.
		
		\vspace{6pt}
	
		\item Step~3: Suppose that $d(x,\aff(w_jw_{j+1})) \leq (1+\eps)\eps\delta$, $d(w_j,\ell_j) \leq 2\eps\delta$, and $d(w_{j+1},\ell_{j+1}) \leq 2\eps\delta$.  Then, we check the cells in $G(x \oplus B_{2\eps\delta})$ in $O(\eps^{-d})$ time to find one that intersects $\aff(w_jw_{j+1})$.  Let $\gamma$ be this cell.   We do not know if $\gamma$ belongs to $\mathcal{G}_1$ or not.  Nevertheless, since $x$ is a grid vertex of $\mathcal{G}_1$, $\gamma$ is within a distance $(1+3\eps)\delta$ from some input curve vertex.  Therefore, $\gamma$ must be a cell in $\mathcal{G}_3$.  There are three cases depending on the relative positions of $w_j$ and $\gamma$.
	
		\vspace{-10pt}
	
		\begin{itemize}
		\item Step~3(a): $w_j \in \gamma \cap\aff(w_jw_{j+1})$.  We claim that $\mathcal{G}_1 \cap G(w_j \oplus B_{7\eps\delta})$ is non-empty, and we report an arbitrary cell in it as the answer for the $(11\eps\delta)$-segment query.  This step takes $O(\eps^{-d})$ time.
		
		\vspace{4pt}
		
		\item Step~3(b): $w_j$ precedes $\gamma \cap \aff(w_jw_{j+1})$ along $\aff(w_jw_{j+1})$ oriented from $w_j$ to $w_{j+1}$.  We query $T_{\ell_j}$ to find the canonical segment $\xi \subset \ell_j$ that contains the projection of $w_j$ in $\ell_j$.  Then, we query $T_\xi$ with $\gamma$ to return $c_{\gamma,\xi}$ as the answer for the $(11\eps\delta)$-segment query.  The time needed is $O(\log\frac{mn}{\eps})$.
		
		
		\vspace{4pt}
		
		\item Step~3(c): $\gamma \cap \aff(w_jw_{j+1})$ precedes $w_j$ along $\aff(w_jw_{j+1})$ oriented from $w_j$ to $w_{j+1}$.  We query $T_{\ell_{j+1}}$ to find the canonical segment $\xi \subset \ell_{j+1}$ that contains the projection of $w_{j+1}$ in $\ell_{j+1}$.  Then, we query $T_\xi$ with $\gamma$ to obtain $c_{\gamma,\xi}$.   
		We claim that $G(c_{\gamma,\xi} \oplus B_{5\eps\delta}) \subset \mathcal{G}_3$ and some cell in  $G(c_{\gamma,\xi} \oplus B_{5\eps\delta})$ intersects $w_jw_{j+1}$.  Pick one such cell $\hat{\gamma}$ in $O(\eps^{-d})$ time.  Either step~3(a) or~3(b) is applicable with $\gamma$ replaced by $\hat{\gamma}$.  Whichever case is applicable, we jump to that case with $\gamma$ replaced by $\hat{\gamma}$ to return an answer for the $(11\eps\delta)$-segment query.  The time needed is $\tilde{O}(\eps^{-d})$.
		
	\end{itemize}
	\end{itemize}
	
	
	\begin{lemma}
	\label{lem:Q-seg}
		It takes $\tilde{O}((mn)^{0.5+\eps}/\eps^d)$ time to answer a $(11\eps\delta)$-segment query.
	\end{lemma}

Lemmas~\ref{lem:D-seg} and \ref{lem:Q-seg} gives the performance of the $(11\eps\delta)$-segment query data structure in Lemma~\ref{lem:segment}.   In Appendix~\ref{app:segment}, we prove the query output correctness in Lemma~\ref{lem:segment}.

\section{Conclusion}

We present $(1+\eps)$-ANN and $(3+\eps)$-ANN data structures that achieve sublinear query times without having space complexities that are proportion to $\min\{m^{\Omega(d)},n^{\Omega(d)}\}$ or exponential in $\min\{m,n\}$.  The query times are $\tilde{O}(k(mn)^{0.5+\eps}/\eps^{O(d)} + k(d/\eps)^{O(dk)})$ for $(1+\eps)$-ANN and $\tilde{O}(k(mn)^{0.5+\eps}/\eps^{O(d)})$ for $(3+\eps)$-ANN.  In two and three dimensions, the query times can be improved to $\tilde{O}(k/\eps^{O(k)})$ for $(1+\eps)$-ANN and $\tilde{O}(k)$ for $(3+\eps)$-ANN.  It is an open problem is to lower the exponential dependence on $d$ and $k$.

\newpage

\bibliography{ref.bib}

\appendix

\section{Proof of Lemma~\ref{lem:D}}

	The performance analysis of $\mathcal{D}$ follows from the previous discussion.  Among the four necessary and sufficient conditions, (i), (iii), and (iv) follow directly from the first, third, and fourth tests.  We show below that the second test and condition (ii) are equivalent.  Clearly, if the second test succeeds, condition (ii) is satisfied.  It remains to analyze the other direction.  By condition (ii), $x''_j \leq_{x_jy_j} y''_j$, $x''_j \in v_{i,a_j} \oplus B_{(1+12\eps)\delta}$, and $y''_j \in v_{i,b_j} \oplus B_{(1+12\eps)\delta}$.  In the second test, we compute the minimum point $x'_j$ in $x_jy_j \cap  (v_{i,a_j} \oplus B_{(1+12\eps)\delta})$ and the maximum point $y'_j$ in $x_jy_j \cap (v_{i,b_j} \oplus B_{(1+12\eps)\delta})$ with respect to $\leq_{x_jy_j}$.  It follows that $x'_j \leq_{x_jy_j} x''_j \leq_{x_jy_j} y''_j \leq_{x_jy_j} y'_j$.  We extend the Fr\'{e}chet matching between $x''_jy''_j$ and $\tau_i[v_{i,a_j},v_{i,b_j}]$ to a matching $\mathcal{M}$ between $x'_jy'_j$ and $\tau_i[v_{i,a_j},v_{i,b_j}]$ by matching $v_{i,a_j}$ with all points in $x'_jx''_j$ and $v_{i,b_j}$ with all points in $y''_jy'_j$.  As $x'_j, x''_j \in v_{i,a_j} \oplus B_{(1+12\eps)\delta}$, convexity implies that $x'_jx''_j \subset v_{i,a_j} \oplus B_{(1+12\eps)\delta}$.  Similarly, $y''_jy'_j \subset v_{i,b_j} \oplus B_{(1+12\eps)\delta}$.  Therefore, $d_F(x'_jy'_j,\tau_i[v_{i,a_j},v_{i,b_j}]) \leq d_\mathcal{M}(x'_jy'_j,\tau_i[v_{i,a_j},v_{i,b_j}]) \leq (1+12\eps)\delta$.

\section{Proof of Lemma~\ref{lem:check}}

	If $x_rx_s \cap (w_{r+1} \oplus B_{(1+\eps)\delta})$ or $x_rx_s \cap (w_s \oplus B_{(1+\eps)\delta})$ is empty, the required $x_r'$ and $x_s'$ do not exist.  Suppose not.
	Let $p$ be the minimum point in $x_rx_s \cap (w_{r+1} \oplus B_{(1+\eps)\delta})$ with respect to $\leq_{x_rx_s}$.  Let $q$ be the maximum point in $x_rx_s \cap (w_s \oplus B_{(1+\eps)\delta})$.   We can check $p \leq_{x_rx_s} q$ and compute $d_F(pq,\sigma[w_{r+1},w_s])$ in $O((s-r)\log(s-r))$ time.  We claim that there is a segment $x'_rx'_s \subseteq x_rx_s$ that satisfies the lemma if and only if $pq$ satisfies the lemma.  The reverse direction is trivial.  The forward direction can be proved in the same way as in the proof of Lemma~\ref{lem:D}.

\section{Proof of Lemma~\ref{lem:query-time}}
	
	We spend $O(\log\frac{mn}{\eps})$ time to obtain $\mathcal{B}[1]$ and $\mathcal{A}[k-1]$.  Then,
	we spend $O(kQ_\text{seg})$ time to obtain $(u_{j,1},u_{j,2})$ for $j \in [k-1]$ and $O(k\log\frac{mn}{\eps})$ search time for each coarse encoding of $\sigma$.   There are $2^{k-3}$ combinations in setting $(c_{j,1},c_{j,2})$ to be $(u_{j,1},u_{j,2})$ or null for $j \in [2,k-2]$.  This gives $2^{k-3}$ possible $\mathcal{C}$'s.  By constraint~3(b), for $j \in [1,k-2]$, we have $\mathcal{A}[j] \in G(c_{j,2} \oplus B_{(1+11\eps)\delta})$, and for $j \in [2,k-1]$, we have $\mathcal{B}[j] \in G(c_{j,1} \oplus B_{(1+11\eps)\delta})$.  Therefore, for each $\mathcal{C}$ enumerated, there are $O(\sqrt{d}/\eps)^{2d(k-2)}$ ways to set $\mathcal{A}$ and $\mathcal{B}$.  In all, the total number of $(\mathcal{A},\mathcal{B},\mathcal{C})$'s enumerated and checked is $O(\sqrt{d}/\eps)^{2d(k-2)}2^{k-3} = O(\sqrt{d}/\eps)^{2d(k-2)}$.  It takes $O(k\log k)$ time to check each by Lemma~\ref{lem:check}.

\section{Proof of Lemma~\ref{lem: distance_bound}}

	Suppose that $T_E \not= \emptyset$ as the lemma statement is vacuous otherwise.
	Take any $\tau_i \in T_E$.  We construct a matching $\mathcal{M}$ between $\tau_i$ and $\sigma$ such that $d_{\mathcal{M}}(\tau_i, \sigma)\le (1+24\eps)\delta$.  Since $T_E \not= \emptyset$, there exists a partition $(\pi_0,\ldots, \pi_{k-1})$ of the vertices of $\tau_i$ that satisfy Lemma~\ref{lem:D}(i)--(iv).
	
	We first match the vertices of $\tau_i$ to points on $\sigma$ as follows. We match $v_{i, 1}$ and $v_{i, m}$ to $w_1$ and $w_k$, respectively. According to constraint 2 and Lemma~\ref{lem:D}(iii), we have $d(v_{i, 1}, w_1)\le (1+\eps)\delta$ and $d(v_{i, m}, w_k)\le (1+\eps)\delta$.   By Lemma~\ref{lem:D}(ii), for $j \in [k-1]$, if $\pi_j\not=\emptyset$ and $\pi_j = \{v_{i,a_j},v_{i,a_j+1},\ldots,v_{a,b_j}\}$, there exist a segment $x_j''y_j''$ between vertices of $c_{j,1}$ and $c_{j,2}$ such that $v_{i, a_j},..., v_{i, b_j}$ can be matched to some points $p_{i,a_j},\ldots,p_{i,b_j} \in x''_jy''_j$ within a distance $(1+12\eps)\delta$ in order along $x_j''y_j''$.  We have $(c_{1,1},c_{1,2}) \not= \text{null}$ by definition, and Lemma~\ref{lem:D}(i) implies that if $j \in [2,k-1]$ and $\pi_j \not= \emptyset$, then $(c_{j,1},c_{j,2}) \not= \text{null}$.  Then,  constraint~1(b) ensure that if $j = 1$ or $\pi_j \not= \emptyset$, there exist two points $z_j \in w_jw_{j+1} \cap (c_{j, 1}\oplus B_{11\eps\delta})$ and $z_j' \in w_jw_{j+1} \cap (c_{j,2}\oplus B_{11\eps\delta})$ such that  $z_j \le_{w_jw_{j+1}} z'_j$.   As $x''_j$ and $y''_j$ are vertices of $c_{j,1}$ and $c_{j,2}$, respectively, we have $d(z_j,x''_j )\le 12\eps\delta$ and $d(z'_j,y''_j) \le 12\eps\delta$.  Therefore, a linear interpolation between $x''_jy''_j$ and $z_jz'_j$ sends $p_{i,a_j},\ldots,p_{i,b_j}$ to points $q_{i,a_j},\ldots, q_{i,b_j} \in z_jz'_j$ within a distance of $12\eps\delta$.  In all, for $l \in [a_j,b_j]$, we can match $v_{i,l}$ to $q_{i,l}$ and $d(v_{i,l},q_{i,l}) \leq d(v_{i,l},p_{i,l}) + d(p_{i,l},q_{i,l}) \leq (1+24\eps)\delta$.  This takes care of the matching of the vertices of $\tau_i$ to points on $\sigma$.
	
	Next, we match the vertices of $\sigma$ to points on $\tau_i$ as follows.  The vertices $w_1$ and $w_k$ have been matched with $v_{i,1}$ and $v_{i,m}$, respectively.  Take a vertex $w_j$ for any $j \in [2,k-1]$.   There is a unique $(r,s) \in \mathcal{J}$ such that $r < j \leq s$.
	If $r = 1$ and $\pi_1 = \emptyset$, let $b_1 = 1$; otherwise, let $b_r = \max\{ b : v_{i,b} \in \pi_r \}$.  By Lemma~\ref{lem:D}(iv), there exist two points $y_r \in \tau_{i,b_r} \cap \mathcal{A}[r]$ and $y_s \in \tau_{i,b_r} \cap \mathcal{B}[s]$ such that $y_r \leq_{\tau_{i,b_r}} y_s$.
	By constraint~3(c), there is a matching of $w_{r+1},..., w_{s}$ to points $z_{r+1},..., z_{s}$ in a segment $x_rx_s$ between vertices of $\mathcal{A}[r]$ and $\mathcal{B}[s]$ such that $z_{r+1}\le_{x_rx_s} z_{r+2}\le_{x_rx_s}\ldots\le_{x_rx_s}z_{s}$, and $d(z_l, w_l) \le (1+\eps)\delta$ for all $l \in [r+1,s]$.  A linear interpolation between $x_rx_s$ and $y_ry_s$ sends $z_{r+1},..., z_{s}$ to points in $\tau_{i,b_r}$ within a distance $\eps\delta$.  Combining this linear interpolation with the matching from $w_{r+1},..., w_{s}$ to $z_{r+1},..., z_{s}$ gives a matching from $w_{r+1},..., w_{s}$ to points in order on $\tau_{i,b_r}$ within a distance $(1+2\eps)\delta$.  This takes care of the matching of the vertices of $\sigma$ to points on $\tau_i$.

	So far, we have obtained pairs of vertices and their matching partners on $\tau_i$ and $\sigma$.  In $\sigma$, the vertices of $\sigma$ and the matching partners of the vertices of $\tau_i$ divide $\sigma$ into a sequence of line segments.  Similarly, the vertices of $\tau_i$ and the matching partners of the vertices of $\sigma$ divide $\tau_i$ into a sequence of line segments.  We complete the matching by linear interpolations between every pair of  corresponding segments on $\tau_i$ and $\sigma$.   The resulting distance is dominated by the distance bound $(1+24\eps)\delta$ of the matching of the vertices of $\tau_i$ to points in $\sigma$.

\section{Proof of Lemma~\ref{lem:query2}}
	
	Let $\sigma_0$ be the polygonal curve produced for $\sigma$.  If the search in $\mathcal{D}$ returns $(\sigma_0,i)$, the construction of $\mathcal{D}$ guarantees that $d_F(\sigma_0,\tau_i) \leq (1+12\eps)\delta$.  Therefore, $d_F(\sigma,\tau_i) \leq d_F(\sigma,\sigma_0) + d_F(\sigma_0,\tau_i) \leq (3+24\eps)\delta$.  It remains to prove that if $\min_{\tau_i \in T} d_F(\sigma,\tau_i) \leq \delta$, then $d_F(\sigma,\sigma_0) \leq (2+12\eps)\delta$ and the search in $\mathcal{D}$ with $\sigma_0$ will succeed.
	
	Let $\tau_i$ be the curve in $T$ such that $d_F(\sigma,\tau_i) \leq \delta$.  Let $\mathcal{M}$ be a Fr\'{e}chet matching between $\sigma$ and $\tau_i$.   Let $((c_{j_r,1},c_{j_r,2}))_{r \in [k_0]}$ be the non-null cell pairs in the construction of $\sigma_0$ from $\sigma$.  By construction, for $j \in [j_r+1,j_{r+1}-1]$, $(c_{j,1},c_{j,2}) = \text{null}$, so $w_jw_{j+1}$ does not intersect any cell in $\mathcal{G}_1$, which implies that $w_jw_{j+1}$ is at distance more than $\delta$ from any input vertex.  Therefore, $\sigma[w_{j_r+1},w_{j_{r+1}}]$ must  be matched by $\mathcal{M}$ to an edge of $\tau_i$, say $\tau_{i,a}$.
	
	By the definition of the $(11\eps\delta)$-segment query, when we walk from $w_{j_r+1}$ to $w_{j_r}$, we must hit $c_{j_r,2} \oplus B_{11\eps\delta}$ at a point, say $x_r$.  Similarly, when we walk from $w_{j_{r+1}}$ to $w_{j_{r+1}+1}$, we must hit $c_{j_{r+1},1} \oplus B_{11\eps\delta}$ at a point, say $y_{r+1}$.  
	
	We claim that $v_{i,a}$ is matched by $\mathcal{M}$ with a point in $\sigma$ that lies does not lie behind $x_r$.  If not, when we walk from $w_{j_r+1}$ to $w_{j_r}$, we must hit a cell in $G(v_{i,a} \oplus B_\delta)$ before reaching $x_r$, but this contradicts the definition of $c_{j_r,2}$ being the output of the $(11\eps\delta)$-segment query with $w_{j_r+1}w_{j_r}$.  Similarly, $v_{i,a+1}$ is matched by $\mathcal{M}$ with a point in $\sigma$ that does not lie in front of $y_{r+1}$.  Figure~\ref{fg:3ann} shows an illustration.
	
	\begin{figure}
		\centerline{\includegraphics[scale=0.5]{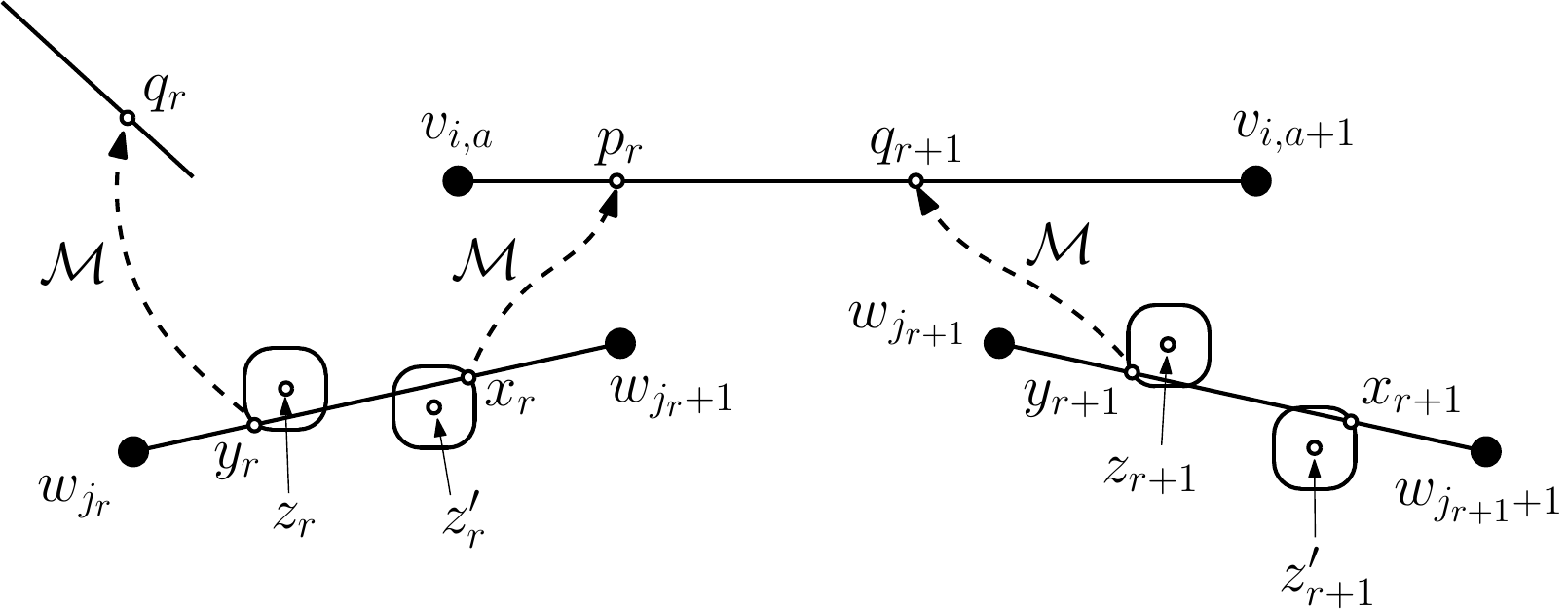}}
		\caption{The curved squares are $c_{j_r,1} \oplus B_{11\eps\delta}$, $c_{j_r,2} \oplus B_{11\eps\delta}$, $c_{j_{r+1},1} \oplus B_{11\eps\delta}$, and $c_{j_{r+1},2} \oplus B_{11\eps\delta}$ from left to right.}
		\label{fg:3ann}
	\end{figure}
	
	Hence, $\mathcal{M}$ matches the subcurve $\sigma[x_r,y_{r+1}]$ to a segment $p_rq_{r+1} \subseteq \tau_{i,a}$.  So $d(p_r,x_r)$ and $d(q_{r+1},y_{r+1})$ are at most $\delta$, which implies that $d_F(p_rq_{r+1},x_ry_{r+1}) \leq \delta$.  Therefore, we can combine $\mathcal{M}$ with a Fr\'{e}chet matching between $p_rq_{r+1}$ and $x_ry_{r+1}$ to show that $d_F(\sigma[x_r,y_{r+1}],x_ry_{r+1}) \leq 2\delta$.  Extending the matching from $x_r$ to the center $z'_r$ of $c_{j_r,2}$ and from $y_{r+1}$ to the center $z_{r+1}$ of $c_{j_{r+1},1}$ gives $d_F(\sigma[x_r,y_{r+1}],z'_{j_r}z^{}_{r+1}) \leq (2+12\eps)\delta$.  Clearly, $d_F(y_rx_r,z_rz'_r) \leq 12\eps\delta$ because $d(y_r,z_r)$ and $d(x_r,z'_r)$ are at most $12\eps\delta$.  As a result, we can combine the Fr\'{e}chet matchings between $y_rx_r$ and $z_rz'_r$ for $r \in [k_0]$ and between $\sigma[x_r,y_{r+1}]$ and $z'_rz^{}_{r+1}$ for $r \in [k_0-1]$ to conclude that $d_F(\sigma,\sigma_0) \leq (2+12\eps)\delta$.  So $\sigma_0$ will pass the check of the query procedure.
	
	We have proved previously that $d_F(p_rq_{r+1},x_ry_{r+1}) \leq \delta$.  Take a Fr\'{e}chet  matching between $p_rq_{r+1}$ and $x_ry_{r+1}$.  We extend it from $x_r$ to $z'_r$ and from $y_{r+1}$ to $z_{r+1}$ to obtain $d_F(p_rq_{r+1},z'_rz_{r+1}) \leq (1+12\eps)\delta$.  Recall that $\mathcal{M}$ matches the subcurve $\sigma[x_r,y_{r+1}]$ to $p_rq_{r+1} \subseteq \tau_{i,a}$.  
	It follows that $\mathcal{M}$ matches $\sigma[y_r,x_r] = y_rx_r$ to $\tau_i[q_r,p_r]$, so $d_F(y_rx_r,\tau_i[q_r,p_r]) \leq \delta$.  As a result, $d_F(z_rz'_r,\tau_i[q_r,p_r]) \leq d_F(y_rx_r,\tau_i[q_r,p_r]) + \max\{d(y_r,z_r), d(x_r,z_r')\} \leq (1+12\eps)\delta$.  We conclude that $d_F(\sigma_0,\tau_i) \leq (1+12\eps)\delta$.  Hence, our preprocessing must have stored an input curve at Fr\'{e}chet distance at most $(1+12\eps)\delta$ with $\sigma_0$, which will be reported.

\section{Proof of Lemma~\ref{lem:C}}
	Consider (i).  The existence of $y \in \gamma$ such that $c \cap xy \not= \emptyset$ implies that $x \in F(c,\gamma)$.  Then, as $x \in L_\xi$, we get $x \in (\aff(\xi) \oplus B_{2\eps\delta}) \cap F(c,\gamma)$.   Also, $\xi$ contains the projection of $x$ because $x \in L_\xi$.  Therefore, as $\xi$ is a canonical segment, the projection of $(\aff(\xi) \oplus B_{2\eps\delta}) \cap F(c,\gamma)$ in $\aff(\xi)$ must contain $\xi$.  It follows that $c \in C_{\gamma,\xi}$ by definition.
	
	Consider (ii).  Let $x$ be any point in $L_\xi$.  Let $y$ be any point in $\gamma$.
	By the definition of $\xi$, there exists a point $x' \in L_\xi \cap F(c,\gamma)$ such that $d(x,x')\leq 4\eps\delta$.  As $x' \in F(c,\gamma)$, there exists $y' \in \gamma$ such that $x'y' \cap c \not= \emptyset$.   We have $d(y,y') \leq \eps\delta$ as both belong to $\gamma$, so a linear interpolation sends every point in $x'y'$ to a point in $x'y$ within a distance $\eps\delta$.  Therefore, $x'y \cap (c \oplus B_{\eps\delta}) \not= \emptyset$.  Similarly, as $d(x,x') \leq 4\eps\delta$, we get $xy \cap (c \oplus B_{5\eps\delta}) \not= \emptyset$.  This proves (ii).   
	
	Assume to the contrary that (iii) is false for $\lambda = 11\eps\delta$.  So there exists $x \in \xi$ and $y \in \gamma$ such that when we walk from $x$ to $y$, we hit a cell $c \in C_{\gamma,\xi}$ strictly before reaching $c_{\gamma,\xi} \oplus B_{11\eps\delta}$.  It follows that if we shrink $c_{\gamma,\xi} \oplus B_{11\eps\delta}$ by a distance $\eps\delta$, it becomes disjoint from $c$.  Also, $c_{\gamma,\xi} \oplus B_{10\eps\delta}$ intersects $xy$ by (ii).  That is, $c$ and $c_{\gamma,\xi} \oplus B_{10\eps\delta}$ are disjoint, and $c \cap xy$ strictly precedes $xy \cap (c_{\gamma,\xi} \oplus B_{10\eps\delta})$ with respect to $\leq_{xy}$.    The disjointness of $c$ and $c_{\gamma,\xi} \oplus B_{10\eps\delta}$ implies that $c \oplus B_{5\eps\delta}$ and $c_{\gamma,\xi} \oplus B_{5\eps\delta}$ are disjoint.  Both $c \oplus B_{5\eps\delta}$ and $c_{\gamma,\xi} \oplus B_{5\eps\delta}$ intersect $xy$ by (i).  As a result, $xy \cap (c \oplus B_{5\eps\delta})$ strictly precedes $xy \cap (c_{\gamma,\xi} \oplus B_{5\eps\delta})$ with respect to $\leq_{xy}$.

	Move the destination of the walk linearly from $y$ to $p_\gamma$.  By (ii), the intersections with $c \oplus B_{5\eps\delta}$ and $c_{\gamma,\xi} \oplus B_{5\eps\delta}$ remain non-empty throughout the move.  Therefore, 
	the intersection order of $c\oplus B_{5\eps\delta}$ and $c_{\gamma,\xi} \oplus B_{5\eps\delta}$ cannot change throughout the move.   Similarly, we move the start of the walk linearly from $x$ to $p_\xi$, and we can conclude that the intersection order of $c\oplus B_{5\eps\delta}$ and $c_{\gamma,\xi} \oplus B_{5\eps\delta}$ does not change throughout the move.  But then we should have preferred $c$ to be $c_{\gamma,\xi}$, a contradiction.
	
	Once we have established that (iii) holds for $\lambda = 11\eps\delta$, it also holds for larger values of $\lambda$ because $c_{\gamma,\xi} \oplus B_\lambda$ can only become bigger.

\section{Proof of Lemma~\ref{lem:segment}}
\label{app:segment}

Consider the query procedure in Section~\ref{sec:segment-query}.

Step~2 works correctly according to the output requirements of the $(11\eps\delta)$-segment query and the $(\kappa,\delta)$-ANN query.  It remains to show that step~3 returns the correct answer.  We first show two technical results.  The first one bounds the distance between $\gamma$ in step~3 and the input vertices.

\begin{lemma}\label{lem:gamma}
	The cell $\gamma$ in step~3 intersects $\aff(w_jw_{j+1})$ and is within a distance $(1+6\eps)\delta$ from some input vertex.  Hence, $\gamma \in \mathcal{G}_3$.
\end{lemma}
\begin{proof}
	Suppose that $\gamma$ is the cell in $\mathcal{G}_3$ computed in step~3 when we come from step~2.  The grid vertex $x$ returned by $D_\mathrm{anp}$ in step~1 is within a distance $(1+\eps)\eps\delta$ from $\aff(w_jw_{j+1})$.   So some cell in $G(x \oplus B_{2\eps\delta})$ must intersect $\aff(w_jw_{j+1})$.  When we come to step~3 from step~2, we must obtain the cell $\gamma \in G(x \oplus B_{2\eps\delta})$ that intersects $\aff(w_jw_{j+1})$.   As $x$ is a grid vertex of $\mathcal{G}_1$, the distance between $\gamma$ and some input vertex is at most $(1+3\eps)\delta$.
	
	The other possibility is that $\gamma$ is $\hat{\gamma}$ obtained in step~3(c) before we jump to step~3(a) or 3(b).  To avoid confusion, we use $\gamma$ to denote the cell in $\mathcal{G}_3$ computed when we come to step~3 from step~2.
	By definition, $c_{\gamma,\xi}$ belongs to $\mathcal{G}_1$ which means that $c_{\gamma,\xi}$ is within a distance $\delta$ from some input vertex $v_{i,a}$.  As $\hat{\gamma} \in G(c_{\gamma,\xi} \oplus B_{5\eps\delta})$, we have $d(v_{i,a},\hat{\gamma}) \leq (1+6\eps)\delta$.  Hence, $\hat{\gamma} \in G(v_{i,a} \oplus B_{(1+6\eps)\delta}) \subset \mathcal{G}_3$.
\end{proof}

\cancel{
	The second result says that if a point $y$ is close enough to a line $\ell \in \mathcal{L}$ and $\gamma$ is the cell in step~3, then any segment that connects $y$ to $\gamma$ must be near all cells in $C_{\gamma,\xi}$, where $\xi$ is the canonical segment that contains the projection of $y$ in $\ell$.
	
	\begin{lemma}\label{lem: g-psi}
		Let $\gamma$ be the cell in $\mathcal{G}_3$.  Let $z$ be any point in $\gamma$.  For any point $y \in \real^d$ and any line $\ell \in \mathcal{L}$,  if $d(y,\ell) \leq 2\eps\delta$, then $yz$ intersects $c \oplus B_{4\eps\delta}$ for all $c \in C_{\gamma,\xi}$, where $\xi$ is the canonical segment on $\ell$ that contains the projection of $y$.
	\end{lemma}
	\begin{proof}
		Let $y'$ be the projection of $y$ in $\ell$.  By the definition of canonical segments, we have $\xi = \bigcap_{c \in C_{\gamma,\xi}} (\ell \oplus B_{2\eps\delta}) \cap F(c,\gamma)$.  Therefore, for every $c \in C_{\gamma,\xi}$, there exists $x_c \in F(c,\gamma)$ such that $d(x_c,y') \leq 2\eps\delta$.  By the definition of $F(c,\gamma)$, there exists $z_c \in \gamma$ such that $x_c z_c$ intersects $c$.  Since both $z$ and $z_c$ belong to $\gamma$, we have $d(z,z_c) \leq \eps\delta$.  Also, $d(x_c,y) \leq d(x_c,y') + d(y,y') \leq 4\eps\delta$.  As $x_c z_c \cap c \not= \emptyset$, we conclude that $yz \cap (c \oplus B_{4\eps\delta}) \not= \emptyset$.
	\end{proof}
}

The second result says that $F(\gamma,c)$ and $F(c,\gamma)$ are disjoint if $\gamma$ and $c$ are.

\begin{lemma}\label{lem:disjoint}
	If $\gamma \cap c =\emptyset$, then $F(\gamma, c)\cap F(c, \gamma)=\emptyset$.
\end{lemma}
\begin{proof}
	Assume to the contrary that $F(\gamma,c)\cap F(c, \gamma)\not=\emptyset$.  Take a point $x \in F(\gamma,c)\cap F(c, \gamma)$.  By the definition of $F(c,\gamma)$, we can find a point $y \in \gamma$ such that $xy \cap c \not=\emptyset$.  Analogously, by the definition of $F(\gamma,c)$, we can find another point $z \in c$ such that $xz \cap \gamma \not= \emptyset$.  Take two points $y' \in xy \cap c$ and $z'\in xz\cap\gamma$.  Observe that $x, y, z, y', z'$ lie in the same plane spanned by $xy$ and $xz$.  So $yz' \cap y'z \not= \emptyset$. Since $c$ is convex and $y',z \in c$, we have $y'z \subset c$.  Analogously, we have $yz' \subset \gamma$.  But then $yz' \cap y'z \subseteq c \cap \gamma$, contradicting the assumption that $c \cap \gamma = \emptyset$.
\end{proof}

Now, we are ready to show that the query procedure either discovers that $\min_{\tau_i \in T} d(\sigma,\tau_i) > \delta$, or returns a correct answer for the $(11\eps\delta)$-segment query.  

\vspace{10pt}

\noindent {\bf Proof of Lemma~\ref{lem:segment}}.  We first argue that if the query procedure returns a cell, it satisfies the requirement of a correct answer for the $(11\eps\delta)$-segment query.  A cell $c_\mathrm{ans}$ is returned either in step~3(a) or 3(b).  If $c_\mathrm{ans}$ is returned in step~3(a), then $c_\mathrm{ans} \oplus B_{7\eps\delta}$ contains $w_j$ by construction, so $c_\mathrm{ans}$ must be a correct answer.  If $c_\mathrm{ans}$ is returned in step~3(b), it is equal to $c_{\gamma,\xi}$.  By the precondition of step~3, $w_j$ projects to $\xi$ and $d(w_j,\xi) \leq 2\eps\delta$.  Therefore, by Lemma~\ref{lem:C}, $c_{\gamma,\xi}$ is the correct answer for the $(11\eps\delta)$-segment query.
	
	It remains to show that if there is a cell in $\mathcal{G}_1$ that intersects $w_jw_{j+1}$, the query procedure either discovers that $\min_{\tau_i\in T} d_F(\sigma,\tau_i) > \delta$ or returns a cell.  Since $\mathcal{G}_1$ contains a cell that intersects $w_jw_{j+1}$ by assumption, $\aff(w_jw_{j+1})$ is within a distance $\eps\delta$ from some grid vertices of $\mathcal{G}_1$, so $D_\mathrm{anp}$ must return a grid vertex $x$ at distance $(1+\eps)\eps\delta$ or less from $\aff(w_jw_{j+1})$.  
	As a result, step~2 cannot return null.  If step~2 returns ``no'' as the answer to the $(\kappa,\delta)$-ANN query, it must be the case that $\min_{\tau_i \in T} d_F(\sigma,\tau_i) > \delta$.  The remaining possibility is that the query procedure proceeds to step~3.
	
	Suppose that step~3(a) is applicable.  In this case, $w_j \in \gamma$, and by Lemma~\ref{lem:gamma}, $\gamma$ is within a distance $(1+6\eps)\delta$ from some input vertex $v_{i,a}$.   Therefore, $w_j$ is at distance $7\eps\delta$ or less from some cell in $G(v_{i,a} \oplus B_\delta) \subset \mathcal{G}_1$.  Hence, step~3(a) will succeed in returning one in $G(w_j \oplus B_{7\eps\delta})$.
	
	If  step~3(b) is applicable, then $c_{\gamma,\xi}$ will be returned.
	
	Suppose that step~3(c) is applicable.  Let $\xi$ be the canonical segment in $\ell_{j+1}$ that contains the projection of $w_{j+1}$ in $\ell_{j+1}$  By Lemma~\ref{lem:C}(ii), $c_{\gamma,\xi} \oplus B_{5\eps\delta}$ intersects $w_{j+1}w_j$.  Since $c_{\gamma,\xi}$ belongs to $\mathcal{G}_1$, it is within a distance $\delta$ from some input vertex $v_{i,a}$, which implies that $c_{\gamma,\xi} \oplus B_{5\eps\delta}$ is covered by the cells in $G(v_{i,a} \oplus B_{(1+6\eps)\delta}) \subset \mathcal{G}_3$.  Therefore, step~3(c) will succeed in finding a cell $\hat{\gamma} \in \mathcal{G}_3$ that intersects $w_jw_{j+1}$.  So either step~3(a) or 3(b) is applicable for $\hat{\gamma}$, and the query procedure will return a cell.  \qed

\end{document}